\documentclass{article}
\usepackage[utf8]{inputenc}

\usepackage[text={6.8in,8in},centering]{geometry}
\usepackage{amsmath}
\usepackage{amssymb}
\usepackage{amsthm}
\usepackage[english]{babel}
\usepackage{graphicx}
\usepackage{subfig} 
\usepackage{tikz}
\usepackage{url}
\usepackage{ulem}
\usepackage{bbold}
\usepackage{hyperref}
\usepackage{todonotes}
\usepackage{authblk}
\usepackage{color}
\renewcommand{\deg}{{\textbf{d}}}

\newtheorem{lemma}{Lemma}[section]
\newtheorem{theorem}[lemma]{Theorem}

\newtheorem{corollary}[lemma]{Corollary }

\newtheorem{claim}{Claim}

\title{Anti-crossings occurrence as exponentially closing gaps\\ in Quantum Annealing}

\author[1,2]{Arthur Braida}
\author[1]{Simon Martiel}
\author[2]{Ioan Todinca}
\affil[1]{Atos Quantum Lab, Les Clayes-sous-Bois, France}
\affil[2]{LIFO - Laboratoire d'Informatique Fondamentale d'Orléans, Université d'Orléans, France}

\date{April 2023}

\begin{document}

\maketitle

\begin{abstract}
    This paper explores the phenomenon of avoided level crossings in quantum annealing, a promising framework for quantum computing that may provide a quantum advantage for certain tasks. Quantum annealing involves letting a quantum system evolve according to the Schrödinger equation, with the goal of obtaining the optimal solution to an optimization problem through measurements of the final state. However, the continuous nature of quantum annealing makes analytical analysis challenging, particularly with regard to the instantaneous eigenenergies. The adiabatic theorem provides a theoretical result for the annealing time required to obtain the optimal solution with high probability, which is inversely proportional to the square of the minimum spectral gap. Avoided level crossings can create exponentially closing gaps, which can lead to exponentially long running times for optimization problems. In this paper, we use a perturbative expansion to derive a condition for the occurrence of an avoided level crossing during the annealing process. We then apply this condition to the MaxCut problem on bipartite graphs. 
    We show that no exponentially small gaps arise for regular bipartite graphs, implying that QA can efficiently solve MaxCut in that case.
    On the other hand, we show that irregularities in the vertex degrees can lead to the satisfaction of the avoided level crossing occurrence condition.
    We provide numerical evidence to support this theoretical development, and discuss the relation between the presence of exponentially closing gaps and the failure of quantum annealing.
    % We support at the end the theoretical development with numerical evidence and question the implication of exponentially closing gaps leading to a failure of quantum annealing.
\end{abstract}

\section{Introduction}

Quantum annealing (QA) is one of the two promising frameworks for quantum computing that may end with a quantum advantage on some specific tasks. Also named adiabatic quantum computing (AQC), it has been introduced by Fahri et. al \cite{farhi2000quantum} in 2000 and stands for the analog part of the gate-based model. Although the two frameworks are known to be equivalent (one can efficiently simulate the other) \cite{albash2018adiabatic}, their studies rely on different theoretical tools. QA has gained lots of attention in the last decade because it seems well-suited to solve combinatorial optimization problems. One largely studied gate-based algorithm, namely QAOA \cite{farhi2014quantum}, is QA-inspired and has brought a lot of attention to the NISQ era. The goal of quantum annealing is to let a quantum system evolve along a trajectory according to the Schrödinger equation. Given some hypotheses, if the Hamiltonians are well defined, measuring the final state after a long enough evolution gives (with high probability) the optimal solution to the optimization problem. This result is guaranteed by the adiabatic theorem. 
This theoretical result describes the annealing time requested to obtain with high probability the optimal solution as a function of the  minimum spectral gap $ \Delta_{\min}$. The latter is defined as the minimum, over the whole adiabatic process, of the difference between the two lowest eigenenergies of the instantaneous Hamiltonian.
%\sout{The adiabatic theorem states that the runtime needs to be inversely proportional to the square of the minimum gap in order to ensure a constant probability of observing the optimal solution.}
%\itodo{Le "needs" me perturbait car le temps long est une condition suffisante, pas nécessaire.} 
The adiabatic theorem states that, by allowing a runtime inversely proportional to the square of the minimum gap, this ensures a constant probability of observing the optimal solution.
In general, exponentially closing minimum gaps yields a running time exponential in the input size but this is only an empirical result.

One major obstacle to this computing model is its analytical analysis, the continuous part of QA makes the equations very difficult to manipulate. The adiabatic theorem has focused a great deal of research on the study of these instantaneous eigenenergies. Since \cite{wilkinson1989statistics}, a physical phenomenon called avoided level crossing (or anti-crossing AC) is known to create an exponentially closing gap bringing the provable runtime to solve an optimization problem to be exponential in the size of the problem. AC is also often referred as first-order quantum phase transition \cite{Amin_2009}. This observation has justified numerous studies of anti-crossings to derive the complexity of quantum annealing runtime. In \cite{altshuler2010anderson}, the authors show that for 
NP-hard problem 3-SAT, an AC will occurs at the end of the evolution, called perturbative crossing, leading to the failure of quantum annealing. However, the appearance of an AC is closely related to the Hamiltonian that one chooses to solve a specific problem. In particular, changing this Hamiltonian can remove or mitigate the avoided crossing \cite{choi2020effects}.
Other authors have focused on giving a more mathematical definition of this phenomenon considering different settings \cite{choi2020effects,Braida_2021}. Finally, in another approach called diabatic annealing, the goal is to manage to create a second AC to compensate for the probability leak of measuring the ground state created by the first AC \cite{feinstein2022effects}. 
%It is important for the reader to note that there is still no consensus about a strict definition of AC other than an exponentially closing gap. It seems that exponentially small gaps can appear while it is not considered an AC and it is unclear that these gaps also lead to a failure of QA. 
It is important for the reader to note that there is still no consensus on a formal definition of an AC. Nevertheless, all definitions agree that an AC only occurs in case of an exponentially closing gap. Some authors \cite{altshuler2010anderson,choi2020effects} consider that exponentially small gaps can appear without constituting an AC because in some situations it is unclear that such gaps come from first-order quantum phase transition.
In the rest of the paper, we will call AC an exponentially closing gap following the work of \cite{Werner:2023zsa}. We will moderate this assertion, as well as the computational inefficiency of QA, in the discussion after the numerical study.

In general, studying the instantaneous eigen pairs, and \textit{a fortiori} ACs, is a hard problem since there is no closed form expression for them. In this work, we use a perturbative expansion of the initial state, the ground state and the first excited state as in \cite{Werner:2023zsa} to manipulate simpler expressions while still being able to say something about the eigenenergies. This perturbative analysis allows us to derive a condition on the occurrence of an AC during the process. We then apply this condition to the well-known MaxCut problem, a fundamental combinatorial optimization problem that has numerous applications in various fields, including computer science (Pinter problem) and physics (Ising models) \cite{Barahona1988}. We show that on regular bipartite graphs, there is no appearance of AC during the annealing but if we remove the regularity constraint, we can construct a family of bipartite graphs that satisfies the condition of AC's occurrence. The striking outcome is that exponentially closing gaps can arise while solving MaxCut on a bipartite graph if it is irregular enough. Although MaxCut on bipartite graphs is known to be solved trivially by classical algorithms, there is no formal proof of efficient (or non-efficient) resolution in the quantum regime. To the best of our knowledge, these are the first proven results on MaxCut using QA. To support this theoretical development, we provide numerical analyses of the gaps of small instances, demonstrating the presence of ACs. 
A final uncommon observation is that QA seems to efficiently solve MaxCut on such instances, despite the presence of exponentially small gaps, thus raising the question of the relation between QA failure and exponential closing gaps.

% A final uncommon observation is the efficiency of QA on such instances despite the small gaps, opening a question on QA failure with exponential closing gap.

\paragraph{Organization of the paper: }The paper is organized as follows. In Section \ref{sec:2}, we introduce the mathematical notations for QA and the preliminaries on the perturbative analysis. In Section \ref{sec:3}, we develop the perturbative analysis on QA from which we derive the condition of occurrence of an anti-crossing. We also show a more practical corollary to use on concrete problems. In Section \ref{sec:4}, we apply the construction to the MaxCut problem on bipartite graphs. First, we show that no AC will appear during annealing if the graph is regular, we then demonstrate that high irregularity can lead to a satisfaction of the AC occurrence condition. We finally construct such a bipartite graph family and we use small instances to plot the eigenvalues and observe the exponentially closing minimum gap from a numerical point of view. We finish with a discussion on AC definition and QA computational efficiency.

\section{Quantum annealing and perturbative analysis}
\label{sec:2}

This section introduces the quantum annealing (QA) framework and explains the general concepts of perturbative analysis.
%In this section we introduce the QA framework, then we show how to develop a perturbative analysis on a system undergoing an annealing process. 

\paragraph{Computing with quantum annealing: }In QA, the instantaneous physical system is represented by the vector $|\psi(t)\rangle$, where the time parameter $t$ goes from 0 to $T$, the runtime of the process. The evolution of this state is governed by the Schrödinger equation subject to a time-dependent Hamiltonian $H(t)$:
$$
i\frac{d}{dt}|\psi(t)\rangle =H(t)|\psi(t)\rangle
$$
where $\hbar$ is taken as unity. The initial state $|\psi_0\rangle=|\psi(t=0)\rangle$ is taken to be the ground state, i.e. the state of minimal eigenvalue, of the initial Hamiltonian $H_0=H(t=0)$. So $H_0$ needs to be easy enough to be able to prepare the initial state. Then the Hamiltonian is smoothly changed toward the final one $H_1=H(T)$ which encodes the solution of a combinatorial optimization problem in its ground state $|GS\rangle$, in the sense that $|GS\rangle$ corresponds to a classical state encoding the optimal solution $x_{opt}$ of our problem.
%\io{in the sense that $H_1$ is chosen such that its ground state corresponds to the optimal solution $x_{opt}$ of our problem}. 
The time-dependent $H(t)$ can be viewed as an interpolation $(1-s(t))H_0+s(t)H_1$ where $s(t)$ denotes the time trajectory going from 0 at $t=0$ to 1 at $t=T$. For a standard linear interpolation $s(t)=\frac{t}{T}$. It is usual to look at the Hamiltonian and the state vector as a function of $s$ and the Schrodinger equation becomes :
$$
i\frac{d}{ds}|\psi(s)\rangle =TH(s)|\psi(s)\rangle, \hspace{1cm} \text{where } H(s)=(1-s)H_0+sH_1 \text{ for } s\in [0,1].
$$
The restriction on starting from the ground state of the initial Hamiltonian comes from the adiabatic theorem. In its more general form, it stipulates that for a ‘‘long enough'' runtime a quantum state $|\psi(t)\rangle$ under a Hamiltonian $H(t)$ stays in the same instantaneous eigenspace during the whole process. Here, ‘‘long enough'' is characterized by the minimum gap $\Delta_{\min}$, namely $T\sim \mathcal{O}\left(\Delta_{\min}^{-2}\right)$. Given that there is a very natural way to encode an optimization problem in a Hamiltonian such that its ground state encodes the solution \cite{farhi2000quantum}, the adiabatic theorem ensures success if the state is initialized in the ground state of $H_0$ motivating the restriction on $|\psi_0 \rangle$. It is important to notice that this restriction is not mandatory if the annealing is out of the adiabatic regime. 

\paragraph{Perturbative analysis: } In general, the perturbative analysis is used to study the effect a perturbation has on a system well defined without. For example, given two Hermitian matrices $A$ and $B$, we know an eigenpair ($x,\lambda$) of $A$, i.e. $Ax=\lambda x$ and we are interested in how a perturbation $B$ will change this state. In other words, if ($x,\lambda$) represents the $k^{th}$ eigenpair of $A$, we are interested in the $k^{th}$ eigenpair ($x_\mu,\lambda_\mu$) of $A+\mu B$ for a small parameter $\mu$. We suppose then that there exists a polynomial expansion in $\mu$ computing ($x_\mu,\lambda_\mu$). We write these expansions as:
\begin{align*}
    x_\mu &= x+ x^{(1)}\mu + x^{(2)}\mu^2+ x^{(3)} \mu^3+ ... \\
    \lambda_\mu &= \lambda +  \lambda^{(1)} \mu+ \lambda^{(2)}\mu^2+ \lambda^{(3)}\mu^3 + ... 
\end{align*} where $x^{(i)}$ and $\lambda^{(i)}$ represent the different coefficients of the polynomial expansion being respectively vectors and scalars. In practice, to be able to say something interesting, we stop the expansion at some order $i$. The validation of the truncation is justified by the ratio of the $(i+1)^{th}$ term over the $i^{th}$ being small. 

The different coefficients are derived iteratively by identification in the eigen relation of the perturbed matrices. Namely, we identify each term in $\mu^j$ in the relation $(A+\mu B)x_\mu=\lambda_\mu x_\mu$. Finally, the obtained relations for each order in $\mu$ are vector equations. After choosing a right basis for the entire space (usually the eigen vectors of $A$), we project along the different basis vectors each relation. Projecting along $x$  gives the $\lambda^{(i)}$ terms and along others basis vectors gives the different coordinates of the vector $x^{(i)}$. For details of the expressions used in quantum mechanics, we refer the reader to an MIT course \cite{zwiebach:2018}. \\

\noindent In this section, we present two concepts: one is how we can use a quantum evolution to compute a solution to an optimization problem, and the second is how to study the evolution of some variables under small perturbation via a perturbative analysis. In the next section, we apply the perturbation analysis directly to QA and see how this helps us to talk about anti-crossing and qualify their occurrences.

\section{Perturbative Analysis to QA}
\label{sec:3}

In this section, we apply the perturbative analysis presented in the previous section to the quantum annealing process. This idea has already been explored by other authors \cite{altshuler2010anderson,Werner:2023zsa} to derive different results and intuitions about the evolution. The perturbative analysis can be naturally applied at the beginning ($s(0)=0$) and at the end ($s(T)=1$) of the evolution. Typically, we know the diagonalizing basis of the Hamiltonians $H_0$ and $H_1$ which allows us to deduce relevant features about the process. Building on the work of \cite{Werner:2023zsa}, we develop here an expansion of the energy $E_0^I$ of the initial state $|\psi_0\rangle$, i.e. the ground state of $H_0$ and one for the energies $E_{gs}$ of the ground state $|GS\rangle$ and $E_{fs}$ of the first excited state $|FS\rangle$ of $H_1$, supposing that the first excited subspace of $H_1$ is degenerated. We are interested in the occurrence of AC which is directly related to the behavior of the instantaneous eigenenergies. Recall that AC refers to the point where the gap is closing exponentially fast, i.e. when the two lowest instantaneous eigenenergies are getting exponentially close to each other. Intuitively, the energy curves almost cross but change directions just before. The expansions of the energies are detailed below in the different subsections. 

Let us set the time-dependent Hamiltonian on which we work. We need to define $H_0$, $H_1$ and the trajectory $s(t)$. We choose to stay in the standard setting of QA for solving classical optimization problems defined over the bitstrings of size $n$ where $s(t)=t/T$, $H_0=-\sum_i \sigma_x^{(i)}$ where the sum is over the $n$ qubits of the considered quantum system and $H_1=\text{diag}(E_x)_{x\in \{0,1\}^n}$. $E_x$ is the value for a classical $n-$ bitstring $x$ of the function we want to optimize, i.e. if $C$ is a cost function to minimize, $C(x)=E_x$. We detail an example with MaxCut problem in section \ref{sec:4}. From this setting, we know that $|\psi_0\rangle$ is the uniform superposition over all bitstrings and the associated eigenspace is non-degenerated. We further assume that the ground space of $H_1$ is non-degenerated as well, i.e. $\exists  ! i, E_i=E_{gs}$ where $E_{gs}$ is the ground state energy (i.e. the optimal value of the target problem) while the first excited subspace is degenerate, i.e. $\exists i\neq j, E_i=E_j=E_{fs}>E_{gs}$, with $E_{fs}$ being the value of the first eigenenergy of $H_1$ above $E_{gs}$. 

We now introduce different graphs that help us to better visualize some quantities. As defined above, $H_0$ can be seen as the negative adjacency matrix of an $n-$regular graph. If each node represents a bitstring $x$, this state is connected to another one $y$ via $H_0$ if $y$ is exactly one bitflip ($\sigma_x$ operation) away from $x$. For any bitstring of size $n$, there are exactly $n$ possible bitflips. $-H_0$ represents the search graph which is the hypercube in dimension $n$ among all possible solutions $x$.  We can isolate the nodes that belong to the degenerated first excited subspace of energy $E_{fs}^T$ among all $x$, i.e. $\text{Loc}=\{y \in \{0,1\}^n|E_y=E_{fs} \}$ and we can define the graph induced by those states $\text{Loc}$ in $-H_0$. We call $G_{loc}$ this subgraph that corresponds to the local minima of the optimization problem. An example of $G_{loc}$ in the 5-cube is shown on Figure \ref{fig:gh0_cycle}. We use MaxCut on a cycle to generate this example, we give the details in the next section and in Appendix \ref{ssec:cycle}. To visualize the landscape of such a graph, we draw in Figure \ref{fig:Cf_cyle} a schematic 2D plot of the objective function $C(x)$ which is also the energy landscape of $H_1$. 
%The term schematic refers to the fact that we cannot actually draw the $x$ on a single axis as they are related via the graph $-H_0$. 
In the example of Figure  \ref{fig:gh0_cycle}, we see that the optimal state $x_{opt}=|GS\rangle$ is entirely linked to $G_{loc}$ and there is no component of $G_{loc}$ far from it, i.e. with a potential barrier in between. This idea is conveyed in Fig \ref{fig:Cf_cyle} by the absence of green parts between the red and blue sections.

\begin{figure}
    \centering
    \includegraphics[scale=0.6,trim={0 3cm 0 2.8cm},clip]{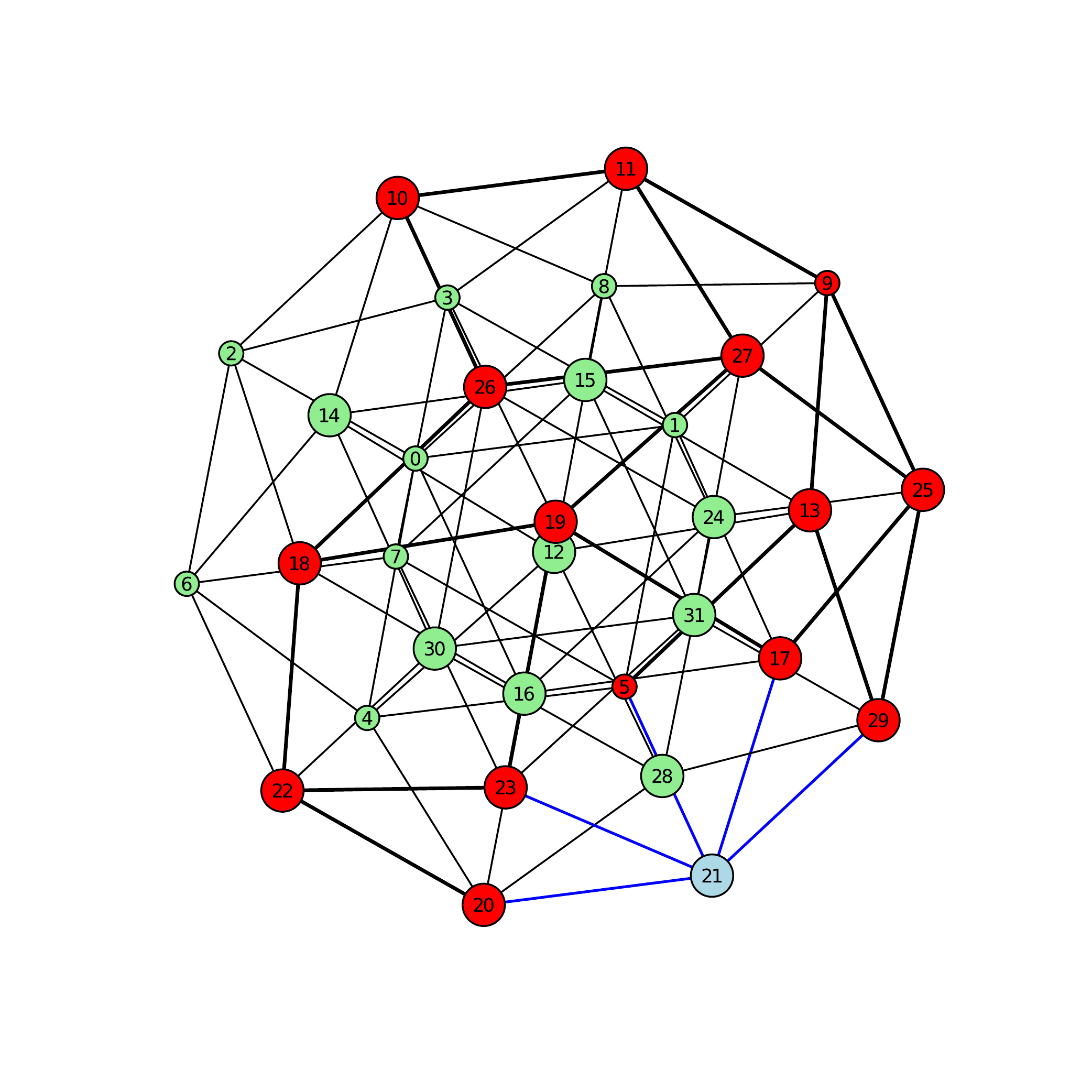}
    \caption{A 5-cube with $G_{loc}$ highlighted with red nodes and thick black edges. Lightblue node is the unique ground state and blue edges show the connection between $G_{loc}$ and the ground state. Green nodes are all the other possible states with higher energies. The labels, once converted in binary, represent the state configuration.}
    \label{fig:gh0_cycle}
\end{figure}

\begin{figure}
    \centering
    \includegraphics[scale=0.6]{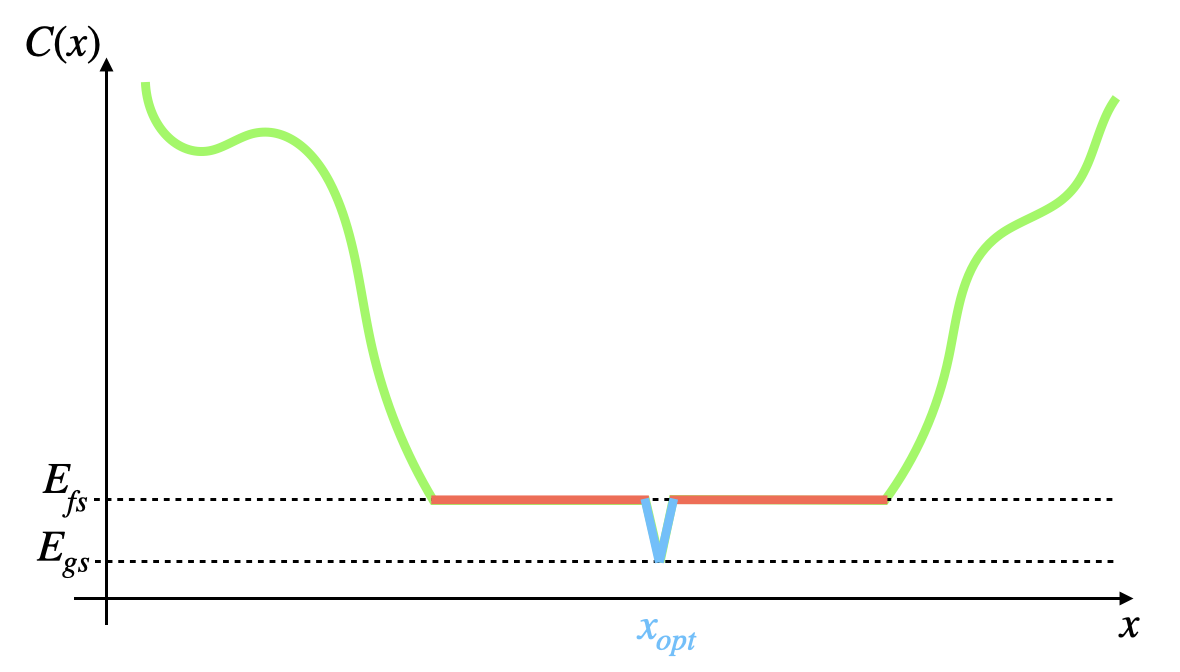}
    \caption{Schematic energy landscape of $H_1$ corresponding to Figure  \ref{fig:gh0_cycle}. $G_{loc}$ has only one component and is strongly connected to the ground state $x_{opt}$.}
    \label{fig:Cf_cyle}
\end{figure}

In the rest of the section, we detail the perturbation expansions and how we can articulate them to derive a condition on the occurrence of the anti-crossing during quantum annealing. More precisely, we will prove the following theorem:
\begin{theorem}\label{th:main} 
Under perturbative expansion validity, if $\lambda_0(\text{loc})$ is the largest eigenvalue of the adjacency matrix of $G_{loc}$ and $H_1$ has a unique ground state and a degenerated first eigenspace, we use a linear interpolation between $H_0$ and $H_1$ as defined above, then by defining 
\begin{align*}
    %s_{dg} &= \frac{n}{n+\langle H_1 \rangle_0 -E_{gs}} \\
    %s_{dl} &= \frac{n-\lambda_0}{n-\lambda_0+\langle H_1 \rangle_0 -E_{fs} } \\
    s_{lg} &= \frac{\lambda_0(\text{loc})}{\Delta H_1 + \lambda_0(\text{loc})} = \frac{1}{1+\frac{\Delta H_1}{\lambda_0(\text{loc})}}\\ 
\end{align*}
and $$\alpha_T=\frac{\Delta H_1}{\langle H_1 \rangle_0 -E_{gs}}$$ where $\Delta H_1 =E_{fs} -E_{gs}$ and $\langle H_1 \rangle_0$ is the mean of $H_1$'s eigenvalues, we can say that an anti-crossing happens at $s_{lg}$ if $\lambda_0(\text{loc}) > n\alpha_T$. No anti-crossing occurs if $\lambda_0(\text{loc}) < n\alpha_T$.
\end{theorem}

This forms a general condition on the occurrence of an anti-crossing during a quantum process with the assumptions of the theorem. We see that the $\alpha_T$ parameter depends only on the problem $H_1$ while $\lambda_0(\text{loc})$ is mixing $H_0$ and $H_1$. We observe from this result that the potential occurrence time of an AC around $s_{lg}$ is ruled by the ratio $\frac{\Delta H_1}{\lambda_0(\text{loc})} $. In practice, this result can help computer scientists to design appropriate schedules by slowing the evolution around the AC. However, the $\lambda_0(\text{loc})$ parameter can be complicated to compute. It encodes the centrality of $G_{loc}$ and can be interpreted as the importance of the graph. To tackle this we use a result from graph theory \cite{graphEig} that bounds the largest eigenvector of a graph by : $\deg_{\text{avg}}(\text{loc}) \leq \lambda_0(\text{loc}) \leq \deg_{\max}(\text{loc})$. Where $\deg_{\text{avg}}(\text{loc})$ and $\deg_{\max}(\text{loc})$ denote the average and maximum degree of $G_{loc}$ respectively. We can derive the following more practical corollary:

\begin{corollary}
By introducing, 
\begin{align*}
    %s_{dg} &= \frac{n}{n+\langle H_1\rangle_0 -E_{gs}^T} \\
    %s_{dl}^+ &= \frac{n-\deg_{\text{avg}}(\text{loc})}{n-\deg_{\text{avg}}(\text{loc})+\langle H_1\rangle_0 -E_{fs}^T }\\
    %s_{dl}^- &= \frac{n-\deg_{\max} (\hbox{loc})}{n-\deg_{\max} (\hbox{loc})+\langle H_1\rangle_0 -E_{fs}^T } \\
    s_{lg}^+ &= \frac{\deg_{\text{avg}}(\text{loc})}{\Delta H_1 + \deg_{\text{avg}}(\text{loc})} \\
    s_{lg}^- &= \frac{\deg_{\max} (\text{loc})}{\Delta H_1 + \deg_{\max} (\hbox{loc})} 
\end{align*}
we can distinguish three regimes :
\begin{enumerate}
    \item[-] AC occurs in the interval $[s_{lg}^+,s_{lg}^-]$ if $\deg_{\text{avg}}(\text{loc}) > n\alpha_T$;
    \item[-] NO-AC occurs if $\deg_{\max}(\text{loc}) < n\alpha_T$;
    \item[-] UNDEFINED if $ \deg_{\max} (\text{loc})>n\alpha_T> \deg_{\text{avg}} (\text{loc})$.
\end{enumerate} 
\end{corollary}

This corollary gives an interval where an AC may occur. Furthermore, it will help anyone who wants to study the different regimes when applying to a specific problem as we do with MaxCut in the next section. In any case, this analytical result is derived from the perturbative theory and the validity of the truncation used needs to be checked. We suggest a validation for MaxCut in Appendix \ref{app:val}. Now let us detail the proof of the theorem. \\

%\noindent \textbf{NB: }Furthermore, in most cases, $\Delta_E^T = \mathcal{O}(1)$, so if $\deg(\text{loc}) \rightarrow \infty$ for large $n$, the position $s_{lg}$ goes to 1. 

%\noindent \textbf{NB:} The case where $G_{loc}$ is composed of only isolated nodes can be disgarded as it means it is as hard to find a local minima as the ground state. This situation does not create anti-crossings. \\ 

\subsection{Initial perturbation}
At the beginning of the evolution, we know that we start from the ground state of $H_0$ with energy $E_0^I$, i.e. $H_0|\psi_0\rangle =E_0^I|\psi_0\rangle$. We are interested in how it changes while perturbing $H_0$ with some $H_1$. More formally, let us look at the modified Hamiltonian $\tilde{H}(\varepsilon)=H_0+\varepsilon H_1$ which is obtained by dividing the original Hamiltonian by $(1-s)$ and setting $\varepsilon = \frac{s}{1-s}$. If we call $E_{\text{deloc}}(\varepsilon)$, 'deloc' for delocalized state, the ground state energy of $\tilde{H}(\varepsilon)$, by perturbative analysis with non-degenerated subspace, the first-order expansion is :
\begin{align*}
    E_{\text{deloc}}(\varepsilon) &= E_0^{(0)} + \varepsilon E_0^{(1)} \\
    &= \langle \psi_0|H_0 |\psi_0 \rangle + \varepsilon \langle \psi_0|H_1 |\psi_0 \rangle \\
    &= E_0^I + \varepsilon \langle H_1 \rangle_0
\end{align*}
where $E_0^I=-n$ and the associated state $|\psi_0\rangle$ is a uniform superposition among all bitstrings. Hence, $\langle H_1 \rangle_0$ represents the mean of all possible values of the optimization problem, encoded in $H_1$. Therefore, in the $s$ frame, we end up with :
\begin{align}
    E_{\text{deloc}}(s)=-(1-s)n+s\langle H_1 \rangle_0
\end{align}

\subsection{Final perturbation}

At the end of the evolution, we know that the ideal case is if the state overlaps largely with the final ground state. However, the occurrence of an anti-crossing may lead to a significant overlap with the first excited state. So we focus our interest on the energy's behavior ending in $E_{gs}$ and $E_{fs}$ while it is perturbed by $H_0$. More formally, let us look at the modified Hamiltonian $\bar{H}(\lambda)=H_1+\lambda H_0$ which is obtained by dividing the original Hamiltonian by $s$ and setting $\lambda = \frac{1-s}{s}$. \\

\noindent We first focus on the behavior of the ground state. We know that $H_1|GS\rangle=E_{gs}|GS\rangle$. If we call $E_{\text{glob}}(\lambda)$, 'glob' for global minima, the ground state energy of $\Bar{H}(\lambda)$, by perturbative analysis with non-degenerated subspace, the first order expansion is: 
\begin{align*}
    E_{\text{glob}}(\lambda) &= E_{gs}^{(0)} + \lambda E_{gs}^{(1)} \\
    &= \langle GS|H_1 |GS \rangle + \lambda \langle GS|H_0 |GS \rangle \\
    &= E_{gs} 
\end{align*} 
Recall that $E_{gs}$ is the optimal value of the optimization problem we look at and the associated eigenspace is non-degenerated. So $|GS\rangle$ is a quantum state that encodes a classical bitstring optimal solution to the problem. In other words, $|GS\rangle$ is a vector of the canonical basis of the Hilbert space and then $\langle GS|H_0 |GS \rangle$ is a diagonal element of $H_0$ which is all 0. Therefore in the $s$ frame, we end up with:
\begin{align}
    E_{\text{glob}}(s) = sE_{gs}^T
\end{align} 

Secondly, we focus on the evolution of the first excited state. However, we supposed that this subspace is degenerated so we need to be more precise about which state we want to study. Let $|FS,k \rangle$ denotes the $k^{th}$ eigenstate of the degenerate eigenspace of $H_1$, by definition $H_1|FS,k \rangle=E_{fs}^T|FS,k \rangle$. If we keep the usual bitstring basis among the degenerated subspace, the first order term $\langle FS,k|H_0 |FS,k \rangle$ will still be 0 and the degeneracy is not lifted. The states $|FS,k \rangle$ can be ordered by continuity of the non-degenerate instantaneous energy landscape of $H(s)$ and thus $\Bar{H}(\lambda)$ also. Therefore we focus on the energy evolution of the state $|FS,0 \rangle$. If we call $E_{\text{loc}}(\lambda)$  the first excited state energy of $\Bar{H}(\lambda)$, by perturbative analysis with non-degenerated subspace, the first order expansion is : 

\begin{align*}
    E_{\text{loc}}(\lambda) &= E_{fs,0}^{(0)} + \lambda E_{fs,0}^{(1)} \\
    &= \langle FS,0|H_1 |FS,0 \rangle + \lambda \langle FS,0|H_0 |FS,0 \rangle \\
    &= E_{fs}^T  + \lambda \langle FS,0|H_0 |FS,0 \rangle
\end{align*} 
 To lift the degeneracy at first-order, we need to find a ``good'' basis $|FS,k \rangle$ for which $ \forall k \geq 1, \langle FS,0|H_0 |FS,0 \rangle < \langle FS,k|H_0 |FS,k \rangle$. We take as basis vectors $|FS,k \rangle$ of the degenerate eigenspace the eigenvectors of $G_{loc}$'s adjacency matrix $A_{loc}$. With this notation, $A_{loc}|FS,k \rangle = \lambda_k |FS,k \rangle$ where we ordered $\lambda_0 > \lambda_1 \geq \lambda_2 \geq ... $ and finally $\langle FS,k|H_0 |FS,k \rangle=-\lambda_k$ by construction. This ensures to lift the degeneracy if the largest eigenvalue of $A_{loc}$ is unique. This happens if $G_{loc}$ has a unique major component which we suppose. Note that if $G_{loc}$ is composed only of isolated nodes, intuitively, they become as difficult as the ground state to find by QA unless there are exponentially of them, we assume from now that this is not the case. Hence, $\lambda_0$ is unique and in the $s$ frame, we end up with:
\begin{align}
    E_{\text{loc}}(s) = sE_{fs}^T -(1-s)\lambda_0
\end{align} 
From \cite{graphEig}, we can bound the largest eigenvector of a graph by : $\deg_{\text{avg}}(\text{loc}) \leq \lambda_0 \leq \deg_{\max}(\text{loc})$, where $\deg_{\text{avg}}(\text{loc})$ and $\deg_{\max}(\text{loc})$ denote the average and maximum degree of $G_{loc}$ respectively. Consequently, we can use the following more practical bounds on $E_{\text{loc}}(s)$:
\begin{align}
    E_{\text{loc}}(s) &\geq sE_{fs}^T -(1-s)\deg_{\max} (\text{loc}) =E_{\text{loc}}^-(s)\\
    E_{\text{loc}}(s) &\leq sE_{fs}^T -(1-s)\deg_{\text{avg}} (\text{loc})=E_{\text{loc}}^+(s)
\end{align}

\subsection{Energy crossing}

\begin{figure}
    \centering
    \includegraphics[scale=0.57]{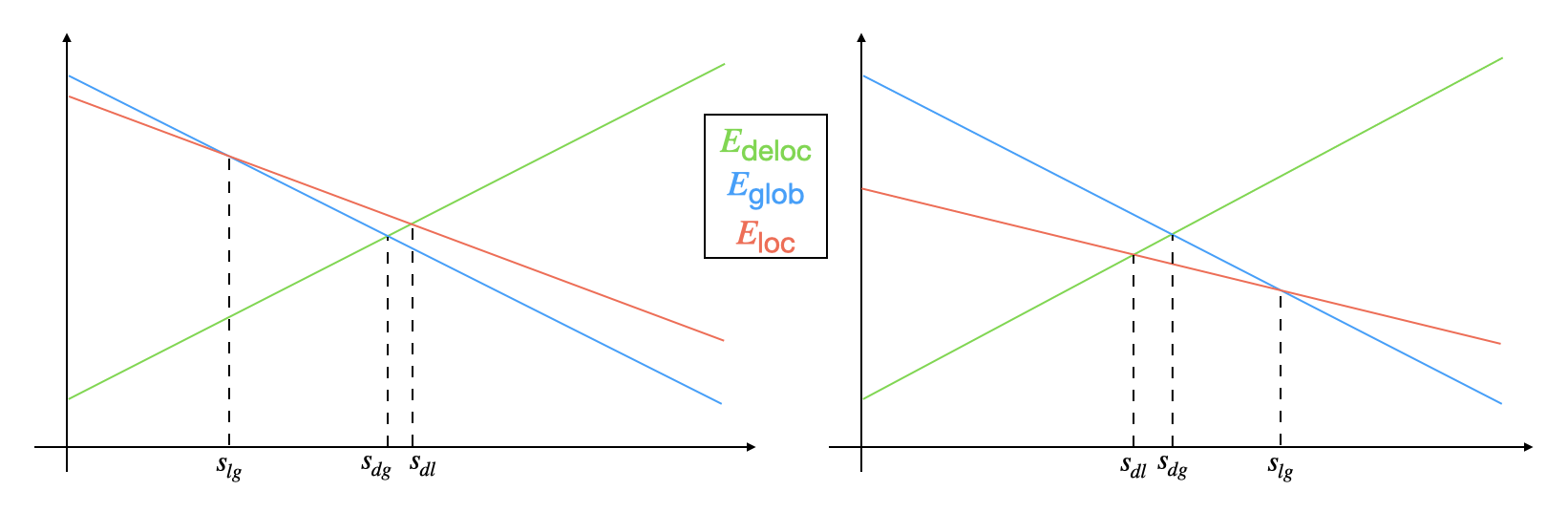}
    \caption{Schematic behavior of the three energy expansions. (left) a case with no AC and (right) case with AC.}
    \label{fig:energy_cross}
\end{figure}
We are set up to distinguish different regimes in which avoided crossing occurs or not. The state starts in the delocalized situation, as $|\psi_0\rangle$ is the uniform superposition, with energy $E_{\text{deloc}}$. If it crosses $E_{\text{glob}}$ first, it then follows the global minima trajectory to ``easily'' reach the final ground state. If it crosses $E_{\text{loc}}$ first, it then follows the local minima trajectory and at some point, it will cross $E_{\text{glob}}$ after and an anti-crossing will occur at this instant. Hence, the two times of interest of the dynamic are $s_{dg}$, defined such that $E_{\text{deloc}}(s_{dg})=E_{\text{glob}}(s_{dg})$, and $s_{dl}$, defined such that $E_{\text{deloc}}(s_{dl})=E_{\text{loc}}(s_{dl})$. If $s_{dl}<s_{dg}$, then an anti-crossing occurs at a time $s_{lg}$ verifying that $E_{\text{loc}}(s_{lg})=E_{\text{glob}}(s_{lg})$. Figure \ref{fig:energy_cross} shows the possible behaviors of the energy expansions. In this plot, we considered that $E_{gs}^T<E_{fs}^T<0$. The slope of the curve $E_{\text{loc}}$ depends in addition on $\lambda_0$, the largest eigenvalue of $G_{loc}$. A larger $\lambda_0$ moves the sign of the slope toward the positive value making $E_{\text{loc}}$ crosses $E_{\text{deloc}}$ before $E_{\text{glob}}$ all others things being equal. This situation (right) will create an AC during the annealing. It is important to note that a large $\lambda_0$ means great connectivity in the graph $G_{loc}$ (or at least in its major component). In other words, it means that the local minima are wide in the mixing graph $H_0$ which makes QA struggle to converge toward the global minima. We understand from this observation that this construction only works if the major component of $G_{loc}$ is not connected to the ground state.

We can derive the explicit expression for $s_{dg},s_{dl}$ and $s_{lg}$ as follow:

% So we need to look at when $E_{\text{deloc}}(s_{dg})=E_{\text{glob}}(s_{dg})$ and $E_{\text{deloc}}(s_{dl})=E_{\text{loc}}(s_{dl})$, and if $s_{dl}<s_{dg}$, then an anti-crossing occurs at $s_{lg}$ verifying that $E_{\text{loc}}(s_{lg})=E_{\text{glob}}(s_{lg})$.
\begin{align*}
    s_{dg} &= \frac{n}{n+\langle H_1 \rangle_0 -E_{gs}^T} \\
    s_{dl} &= \frac{n-\lambda_0}{n-\lambda_0+\langle H_1 \rangle_0 -E_{fs}^T } \\
    s_{lg} &= \frac{\lambda_0}{\Delta E^T + \lambda_0} = \frac{1}{1+\frac{\Delta E^T}{\lambda_0}}\\
\end{align*}
We note $$\alpha_T=\frac{\Delta H_1}{\langle H_1 \rangle_0 -E_{gs}^T}$$ where $\Delta H_1=E_{fs}^T-E_{gs}^T$, a parameter that depends only on the problem $H_1$ we want to solve. And so AC occurs at $s_{lg}$ if $s_{dl} < s_{dg}$ i.e. if $\lambda_0 > n\alpha_T$. This concludes the proof of our theorem. 

The corollary immediately follows by using $E_{loc}^-(s)$ and $E_{loc}^+(s)$. The undefined regime is then when $s_{dg} \in [s_{dl}^-,s_{dl}^+]$ because we cannot discriminate between which curve the delocalized energy will cross first. 

This result is quite general for many targets Hamiltonians, but we still need two conditions: the ground state must be unique and the first excited subspace is degenerated. \\

\noindent In this section, we apply the perturbative analysis to QA and show in the assumption where the ground state of $H_1$ is unique and its first excited subspace is degenerated, that anti-crossings may occur during annealing given a condition to satisfy that depends on $G_{loc}$ and $H_1$. We also give a corollary which relaxed the condition of the theorem to be more useful when applying to a specific problem. In the next section, we show such an application in the case of MaxCut on bipartite graphs.

\section{Application to MaxCut}
\label{sec:4}

In this section, we apply the last theorem to the MaxCut problem. Given a graph $G(V,E)$, the goal of MaxCut is to partition its node set $V$ into two parts $L$ and $R$ in order to maximize the number of cut edges, i.e., of edges with an endpoint in $L$ and the other in $R$. Such partitions are classically encoded by a bitstring of size $n = |V|$, the $i^{th}$ bit being set to $0$ if node $i \in L$, and to $1$ if $i \in R$. We define our target Hamiltonian as $H_1=-\sum_{(ij) \in E}\frac{1-\sigma_z^{(i)} \sigma_z^{(j)}}{2}$. This Hamiltonian (and the corresponding MaxCut cost function) has a trivial symmetry: any solution can be turned into a solution with the same cost by bit-flipping all its entries. Consequently, $H_1$ has a degenerated ground state. We can break down this symmetry by forcing an arbitrary bit (say the first one) to $0$ and updating $H_1$ accordingly. 

To ensure that the two conditions of our theorem are met, we need to choose a class of graphs such that the ground state is non-degenerated (after breaking the trivial symmetry). Connected bipartite graphs obviously respect this property and we focus on them in the rest of the section. We will in particular show that the first excited subspace is degenerated.
Also this class allows us to explicitly determine the parameter and the graph . This will help us to determine the existence (or not) of ACs while solving MaxCut on these graphs with QA

\subsection{d-regular bipartite graphs}

We first restrict the bipartite graphs on being $d$-regular and we will show that no AC appears during the evolution by using the result of the corollary: $\deg_{\max}(\text{loc})<n\alpha_T$. Leading to the following theorem :

\begin{theorem}[NO AC - d-regular bipartite graphs]
\label{thm4}
    Quantum Annealing efficiently solves MaxCut on $d-$regular bipartite graphs.
\end{theorem}

First, we show the two following claims to give a value to $n\alpha_T$, then we show the NO-AC conditions with lemma \ref{lem:NOAC} if $d \notin \{2,4\}$. The latter two cases are detailed in Appendix \ref{app:case} where we directly use the theorem to prove the desired result. 

\begin{claim}
For $d$-regular bipartite graphs we have, $n \alpha_T=\frac{4l}{d}$, where $l\in[1,d]$ denotes the number of uncut edges in the first excited state, i.e. $E_{fs}^T=E_{gs}^T+l$.
\end{claim}
\noindent For bipartite graphs we have that $\langle H_1 \rangle_0=-\frac{|E|}{2}$, $E_{gs}^T=-|E|$ and $\Delta H_1=l \in [1,d]$. For regular graphs, we also have that $|E|=\frac{dn}{2}$. So $n \alpha_T=\frac{4l}{d}$ and we need to look at how $\deg_{\max} (\text{loc})$ and $\deg_{\text{avg}} (\text{loc})$ behave compare to $\frac{4l}{d}$. 

\begin{claim}
\label{claim2}
    There exist graphs with $\deg_{\max} (\text{loc})>0$ only if $l=d$. Therefore $n\alpha_T=4$.
\end{claim}

Recall that $G_{loc}$ is the subgraph induced by solutions of energy $E_{fs}^T$ in the hypercube $-H_0$. In full words, the vertices of  $G_{loc}$ are configurations (bitstrings) of energy $E_{fs}^T$ (so ‘‘second best'' solutions for MaxCut), and two vertices are adjacent if the corresponding bitstrings differ in exactly one bit, i.e., each one is obtained by bit flipping a single bit of the other. We denote by $\deg_{\max}(\text{loc})$ the maximum degree of $G_{loc}$. 
%i.e. we look for the bitstring (a node of $G_{loc}$) with the greatest number of neighbors. A neighbor of a configuration $x$ in $G_{loc}$ is obtained by bitflipping a single bit in $x$. So we look for a configuration $x$ in $G_{loc}$ such that after flipping one bit of $x$, the obtained configuration is still in $G_{loc}$. 
We know that, in the input graph $G$, there exists a partition left/right of its vertices such that all edges lie across the partition (by bipartiteness). 
Looking at one configuration of the first excited subspace, it specifies another bipartition, this time with all but $l$ edges lying across it. 
%\so{j'ai essayé d'améliorer la formulation, toujours pas super satisfait.} 
We are interested in configurations that are not isolated in $G_{loc}$ because these nodes as mentioned in Section \ref{sec:3} do not play a role in AC occurrence.
% If we look at one configuration of the first excited subspace(e.g. on the left we label the node with '0' and '1' on the right)\so{pas beau/clair}, only $l$ edges should be strictly on one side. 
In such a configuration $x$, we want that by flipping one node (i.e. moving it to the other side of the partition), the number of uncut edges stays the same, in order to obtain a configuration $y$ that is also a vertex of $G_{loc}$. So this specific node needs to have half of its edges that is uncut and the other half that is cut in this particular configuration $x$ of the first excited subspace. This automatically restricts $l$ to be both even and larger than $d/2$.

\begin{figure}[ht]
    \centering
    \includegraphics[scale=0.6]{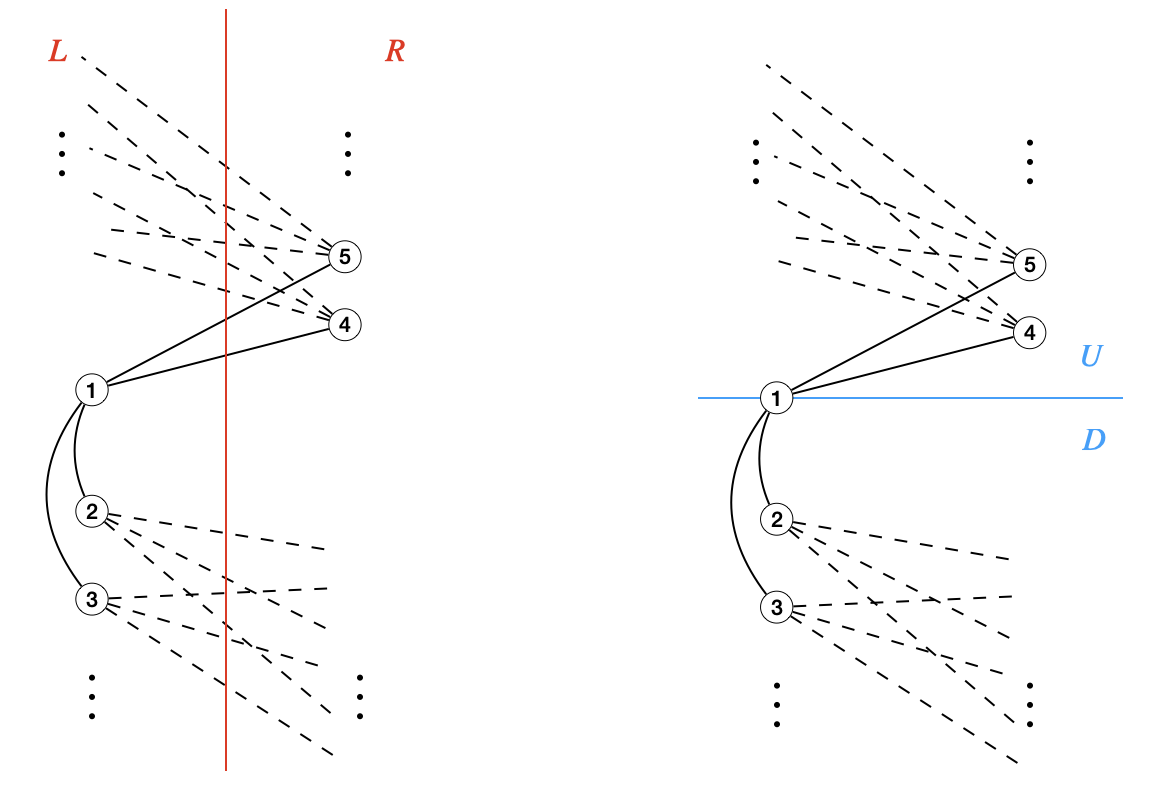}
    \caption{Construction of a specific first excited configuration. The L/R partition (left) is natural in MaxCut. The U/D partition (right) is relevant if 1 is a minimal separator.}
    \label{fig:localconf}
\end{figure}

\paragraph{Case $l=d/2$:} Let us suppose $l=d/2$. We are in a situation similar to Figure \ref{fig:localconf} (left), and see if we can create a bipartite graph from this. By supposing that $l=d/2$, it means, in the configuration of one excited state, all other edges must go from left (L) to right (R). This splits the configuration in the classical L/R partition of a cut.
Then we show the following claim that node 1 is a minimal separator of the graph which creates another split up (U) and down (D) (Fig \ref{fig:localconf} - right).

\begin{claim}
\label{claim3}
    Assume that $l=d/2$ and let us consider a configuration corresponding to a non-isolated vertex of $G_{loc}$. Then there is a node of the input graph $G$, say node 1, having $d/2$ neighbors on each side of the configuration. Moreover, this node is a minimal separator of the graph (see Figure~\ref{fig:localconf}). 
\end{claim}
The configuration $x$ is such that all edges but $l=d/2$ are cut, and this also holds after the bitflip of one of its bits. Assume w.l.o.g. that this is the first bit, corresponding to node 1, and that 1 is on the left-hand side of the configuration, i.e., $1 \in L$. Since flipping node 1 from left to right maintains the number of cut edges, it means that 1 has exactly $d/2$ neighbors in set $L$ and $d/2$ in set $R$. Since $l = d/2$, it also means that the $l$ uncut edges are precisely the $d/2$ ones incident to node $1$, from $1$ to vertices of $L$.

Let $N_D(1)$ denote the set of neighbors of $1$ in $L$, and $N_U(1)$ denote the set of neighbors of $1$ in $R$.
We prove that $N_D(1)$ and $N_U(1)$ are disconnected in graph $G - 1$, obtained from the input graph $G$ by deleting vertex $1$. By contradiction, assume there is a path $P$ from $a \in N_U(1)$ to $b \in N_D(1)$ in  $G - 1$. Path $P$ together with vertex $1$ form a cycle in graph $G$.
%Suppose that 1 is not a minimal separator of the graph then it exists a cycle that goes through node 1, one node in the upper part (eg 4 or 5 in the figure) and one node in the down part (eg 2 or 3). 
By bipartiteness, this cycle is even, so at least one edge of the cycle, other than $\{1,b\}$, is contained in L or R. This is in contradiction with the assumption that $l=d/2$ and all of of these specific $d/2$ edges are linked to the same node $1$. Therefore, $G-1$ is disconnected. This proves claim \ref{claim3}.

This creates four quadrants UL, UR, DL and DR as follows: U is the subset of nodes of $G$ formed by the union of  connected components of $G-1$ intersecting $N_U(1)$, and D is its complement. Then UL, UR, DL and DR are defined as the respective intersections of U and D with L and R (UL = U $\cap$ L and similar). The above considerations tell us that all edges of $G - 1$ go either from UL to UR or from DL to DR. Now, we call $n_{DL},n_{DR}$ the number of nodes in part DL and DR (others than the labelled ones, i.e., the neighbors of node 1). By counting the edges from DL to DR, observe that these variables must satisfy the following equation:
\begin{equation*}
    \frac{d}{2}(d-1) + d n_{DL} = d n_{DR}
\end{equation*}
Because we know that $d$ is even, $n_{DL}$ and $n_{DR}$ are integers, the above equation cannot be satisfied. 

\paragraph{Case $l>d/2$:} $l$ must be strictly larger than $d/2$, i.e. $l \in [\frac{d}{2}+1,d]$. 
All these $l$ uncut edges can be split between $r_L$ and $r_R$, the ones on the left side and right side respectively and wlog we choose that already $d/2$ of them are on the left side. So $r_L \in [\frac{d}{2},d]$, $r_R \in [0, \frac{d}{2}]$ and $l=r_L+r_R$. Again we can count the number of edges that lie across L and R and we end up with:
\begin{equation*}
    dn_L - 2r_L = dn_R - 2r_R 
\end{equation*} where $n_L=|L|$, $n_R=|R|$ and $n_L+n_R=n$ the total number of nodes. In a $d-$regular bipartite graph, $n$ is necessarily even, so we have that 
$$
r_L-r_R = 2(\frac{n}{2}-n_R)\frac{d}{2}=kd \quad \text{ for } k\in \mathbb{Z}
$$ The potential values for $r_L$ and $r_R$ bring the interval for $r_L-r_R$ to $[0,d]$. So only $k=0$ and $k=1$ are possible. If $k=0$, then $r_L=r_R=d/2$ so $l=d$. If $k=1$, then $r_R=0$, $r_L=d$ so $l=d$. In any case, the only possibility is to have $l=d$ which concludes the proof of claim \ref{claim2}. \\ 

These two claims simplify the expression of the different AC occurrence conditions, becoming:
\begin{enumerate}
    \item[-] AC if $\deg_{\text{avg}}(\text{loc}) > 4$;
    \item[-] NO-AC if $\deg_{\max}(\text{loc}) < 4$;
    \item[-] UNDEFINED if $ \deg_{\max} (\text{loc})>4> \deg_{\text{avg}} (\text{loc})$.
\end{enumerate} 

We are left with a last thing to show to assure that no AC occurs while solving MaxCut on d-regular bipartite graphs with QA. To this purpose, we show this final lemma:

\begin{lemma}
\label{lem:NOAC}
    If $d \notin \{2,4 \}$ then $ \deg_{\max} (\text{loc}) < 4$.
\end{lemma}
\begin{proof}
    Recall that odd values for $d$ are already disregarded as $d$ must be even. Suppose it is possible to have $d_{\max}(\text{loc}) \geq 4$ then it means that we need at least 4 nodes in a configuration such as Figure~\ref{fig:localconf}, where half of their edges are uncut. Let us call $F$ the set of these latter nodes, i.e. $|F|=d_{\max}(\text{loc})$. It means that there are at least $|F|*\frac{d}{2}\geq 2d$ outgoing uncut edges from the nodes in $F$. By outgoing edge from a node, we mean the extremity of the edge that leaves the node (each edge contributes to two outgoing edges, one for each of its nodes). So here we count the number of edges that leave a node in $F$ which are uncut. We are allowed to at most $d$ uncut edges to be a local minimum. So all of these $2d$ outgoing uncut edges need to generate exactly $d$ edges. This remark forces $d_{\max}(\text{loc})$ to be smaller than 4, so suppose $d_{\max}(\text{loc}) = 4$. One node has only 3 possible neighbors for its $d/2$ uncut edges, so it is possible as long as $\frac{d}{2} \leq 3$, i.e., $d\leq 6$. For $d=6$, linking all of these edges creates a triangle which makes the whole graph non-bipartite.
%
 %   \io{Proposition de texte alternatif, avec un peu plus de détails.}
%    
%     Recall that odd values for $d$ are already disregarded as $d$ must be even. Suppose it is possible to have $d_{\max}(\text{loc}) \geq 4$ then it means that we need at least 4 nodes in configuration such as Figure~\ref{fig:localconf}, where exactly half of their incident edges are uncut. Let $F$ denote the set of these vertices, so $|F| = d_{\max}(\text{loc})$. Recall that the total number $l$ of uncut edges is at most $d$.
%
%     Firstly we show that $|F|\leq 4$. Assume by contradiction that $|F|\geq 5$. In this case, in the bipartion of graph $G$, at least three vertices of $F$ are on the same part, in particular they are pairwise non-adjacent. Therefore the edges incident to these vertices are pairwise distinct. Since each such vertex has $\frac{d}{2}$ uncut incident edges, the total number of uncut edges would be at least $\frac{3d}{2} > d$, a contradiction.
%
%     Assume now that $d_{\max}(\text{loc}) = 4$. With the same arguments as above, there are two vertices of $F$, say $a$ and $b$, in the same part of $G$, in particular they are non-adjacent. Thus the uncut edges incident to $a$ and $b$ are pairwise distinct, and they form a set of uncut edges of size $d$. We deduce that $l=d$ all uncut edges must be in this set. Therefore the two other vertices of $F$ are on the opposite part of $G$ w.r.t.  $a$ and $b$. Moreover their incident uncut edges are also incident to either $a$ or $b$. Thus $d \leq 4$, which concludes the proof of the lemma.
\end{proof}
\paragraph{Case $d=2$ and $d=4$.}  In these two cases, $d_{\max}(\text{loc}) = 4$ and $d_{\text{avg}}(\text{loc}) < 4$, so they fall in the UNDEFINED regime and further studies are necessary. In Appendix \ref{app:case}, we detail how we can still classify them in the NO-AC regime by directly using the more technical result from the theorem.  \\

\noindent These above results allow us to conclude on the absence of anti-crossing during an annealing process to solve MaxCut d-regular bipartite graph for $d \notin \{2,4 \}$ and show theorem \ref{thm4}. One can deduce from this that there is no exponentially closing gap leading to a polynomial runtime to find the optimal cut in regular bipartite graphs via QA. A natural question rises from this conclusion: can we draw a similar conclusion for general bipartite graphs? We discuss this in the next subsection.

\subsection{General bipartite graphs}\label{sub:gen_bip}

In this section, we are interested in the behavior of the energies if we look at bipartite graphs in general. We construct a family of bipartite graphs that respect the condition of occurrence of an anti-crossing, meaning that exponentially closing gaps can arise even for MaxCut on bipartite graphs. Let $G(E,V)$ denotes a bipartite graph. Similarly to the previous section, $\langle H_1 \rangle_0=-\frac{|E|}{2}$, $E_{gs}=-|E|$ and $\Delta H_1=l \in [1,\deg_{\min}(G)]$. Claim \ref{claim2} is still applicable with the minimum degree $\deg_{\min}(G)$ of $G$. So $\Delta H_1=\deg_{\min}(G)$ and $n\alpha_T$ becomes $4\frac{\deg_{\min}(G)}{\deg_{\text{avg}}(G)}$. The condition for the different regimes can be written as:
\begin{enumerate}
    \item[-] AC if $\deg_{\text{avg}}(\text{loc}) > 4\frac{\deg_{\min}(G)}{\deg_{\text{avg}}(G)}$;
    \item[-] NO-AC if $\deg_{\max}(\text{loc}) < 4\frac{\deg_{\min}(G)}{\deg_{\text{avg}}(G)}$;
    \item[-] UNDEFINED if $ \deg_{\max} (\text{loc})>4\frac{\deg_{\min}(G)}{\deg_{\text{avg}}(G)}> \deg_{\text{avg}} (\text{loc})$.
\end{enumerate} 
The first point gives us the condition for a graph $G$ that produces an anti-crossing under a QA evolution for the MaxCut problem. Firstly, looking only at the right-hand side, the ratio $\frac{\deg_{\min}(G)}{\deg_{\text{avg}}(G)}$ is small for highly irregular graphs. From what we have seen in the previous subsection, the average degree for $G_{loc}$ is certainly smaller than 4 so we need to play with the degree of $G$. 
Even though we remove the regularity hypothesis, we can still use some results from the above cases. Indeed, 
in that setting, $G_{loc}$ arises from the bi-partition of a $\deg_{\min}(G)$-regular induced subgraph of $G$.
%\io{if we look at the part of the graph that generates $G_{loc}$, it is as a $\deg_{\min}(G)$-regular graph, in the sense that all nodes that can be flipped are of minimum degree in $G$.}
We look at graphs $G$ with a large average degree but  with also a small minimum degree and a large $\deg_{\text{avg}}(\text{loc})$. The cycle produces the densest $G_{loc}$ but it is highly connected to the ground state and the average degree of the cycle is not quite large. The idea is to attach two complete bipartite graphs ($K_{rr}, K_{ll}$) that will increase the average degree of the graph by two parallel sequences of nodes of degree 2 ($P_1, P_2$) that will create the dense $G_{loc}$ and small $\deg_{\min}(G)$, equals to 2. Figure \ref{fig:graphconf} provides an example of a such graph with $r=l=3$ and $P_1,P_2$ are sequences of $k_1=k_2=2$ adjacent nodes of degree 2. $k_1$ and $k_2$ need to be of the same parity to assure bipartiteness of the whole graph. Three configurations of the same graph are shown, corresponding to the ground state (left), and two configurations of the first excited subspace (middle, right), that create the different components in $G_{loc}$~(Fig \ref{fig:gloc}). \\ \\

\begin{figure}[ht]
    \centering
    \includegraphics[scale=0.6]{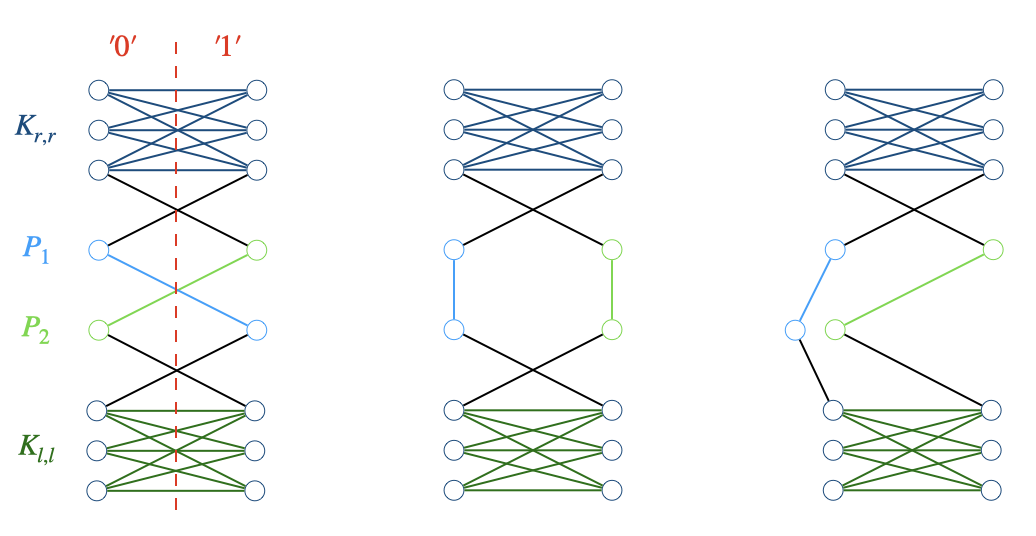}
    \caption{Configurations of $G$ in its ground state (left) and first excited state. (middle) is a configuration far from GS, (right) a configuration neighboring GS.}
    \label{fig:graphconf}
\end{figure}

\begin{figure}[ht]
    \centering
    \includegraphics[scale=0.6]{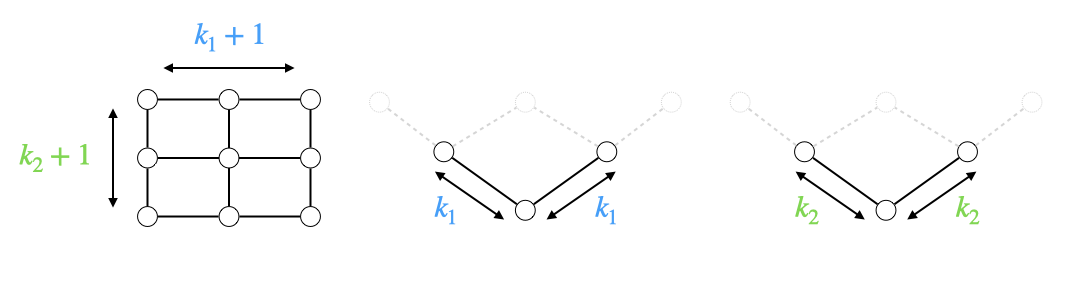}
    \caption{$G_{loc}$ of graph $G$. Three components : (middle) and (right) similarly: components corresponding to states in a configuration close to the one on (Fig \ref{fig:graphconf} - right) and (left) component corresponding to states in the configuration of (Fig \ref{fig:graphconf} - middle). The light dashed grey edges and nodes show how these two components grow when $k_i>2$.}
    \label{fig:gloc}
\end{figure}
\noindent The largest component of $G_{loc}$ is a lattice of size $(k_1+1,k_2+1)$ if $k_i$ represents the number of nodes in $P_i$. It is far away from the ground state as we need to flip at least all the nodes of the $K_{r,r}$ part. The two other components can be viewed as subgraphs of the large component so they have smaller eigenvalues than the largest one of the lattice; they are also strongly connected to the ground state. Figure \ref{fig:detailedGloc} shows the details of the relation between the nodes of $G_{loc}$ and graph configurations in a left/right partition. The middle configuration of Figure \ref{fig:graphconf} corresponds to the middle node of the lattice in $G_{loc}$. Then moving each node in blue or green produce another configuration with the same edge penalty.

\begin{figure}[ht]
    \centering
    \includegraphics[scale=0.7]{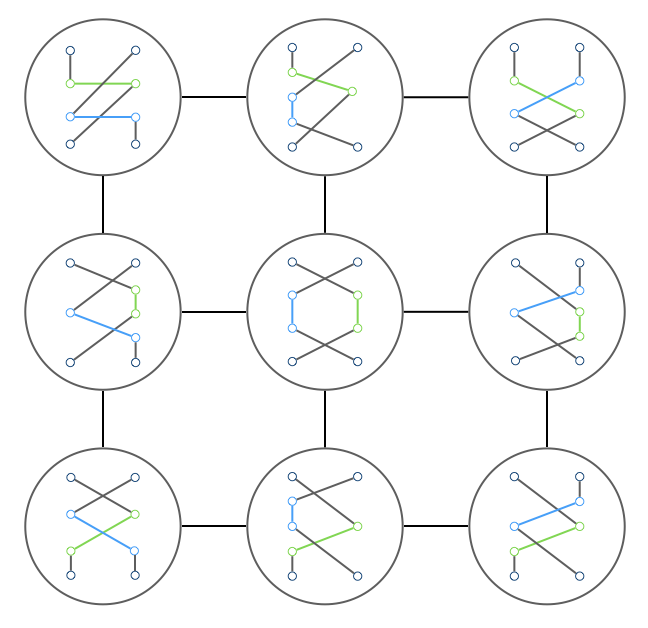}
    \caption{Details of the large component of $G_{loc}$ and how each configuration is related by bitflip. We intentionally omit the drawing of the $K_{r,r}$ and $K_{l,l}$ which do not play a role in 
    $G_{loc}$.}
    \label{fig:detailedGloc}
\end{figure}

We directly have that $\deg_{\min}(G)=2$. Now, we need to derive the average degree of $G$ and of the largest component of $G_{loc}$ (its maximum degree being 4).

\begin{align*}
    \deg_{\text{avg}}(\text{loc}) &= \frac{4*2 + 2(k_1 -1 + k_2 -1)*3 +(k_1 -1)(k_2 -1)*4}{(k_1+1)(k_2+1)} \\
    &= \frac{4k_1k_2 +2(k_1+k_2)}{(k_1+1)(k_2+1)} \\
    &= 4 \left (1-\frac{1+\frac{1}{2}(k_1+k_2)}{(k_1+1)(k_2+1)} \right) \\
    &= 4 \left (1-\frac{1}{k+1} \right) \hspace{6.5cm} \text{ for } k=k_1=k_2 \\
    \deg_{\text{avg}}(G) &= \frac{(k_1+k_2)*2 + 2r*r +2l*l+4}{k_1+k_2+2r+2l}\\
    &=\frac{2k+r^2+l^2+2}{k+r+l} \hspace{6.27cm} \text{ for } k=k_1=k_2 
\end{align*}

Let's solve the equation $\deg_{\text{avg}}(\text{loc}) > 4 \frac{\deg_{\min}(G)}{\deg_{\text{avg}}(G)}$ with $\deg_{\min}(G)=2$.
\begin{align*}
    \deg_{\text{avg}}(\text{loc}) &> 4 \frac{\deg_{\min}(G)}{\deg_{\text{avg}}(G)} \\
    \Rightarrow 1-\frac{1}{k+1}  &> \frac{2(k+r+l)}{2k+r^2+l^2+2} \\
    \Rightarrow \frac{r^2+l^2+2-2r-2l}{2k+r^2+l^2+2}  &> \frac{1}{k+1} \\
    \Rightarrow (k+1)(r^2+l^2+2-2r-2l)   &> 2k+r^2+l^2+2 \\
    \Rightarrow k(r^2+l^2-2r-2l)   &> 2r+2l \\
    \Rightarrow k   &> \frac{2(r+l)}{r(r-2)+l(l-2)} \\
    \Rightarrow k   &> \frac{2(r+3)}{r(r-2)+3} \hspace{3.1cm} \text{ for } l=3 
\end{align*} We have a limit at $r=3$ and $k=2$ for a graph of size 16. Then the smallest graphs that satisfy the condition are for $r=3$ and $k=3$ or $r=4,l=3$ and $k=2$ which bring the size of the smallest graphs satisfying AC condition to 18 nodes. \\

\noindent This above construction shows that there exist bipartite graphs that exhibit an AC. The presence of an anti-crossing implies an exponentially closing gap bringing the provable runtime to find the optimal cut exponentially large in the size of the graph. This construction can be scaled up easily by growing the parameters $k,r$ and $l$. In the next subsection, we numerically investigate the presence of AC on graphs of this family to support this theoretical result. 

\subsection{Numerical study: AC and other observations}
In this section, we give some numerical evidence of the occurrence of the AC in the particular family we constructed in Subsection \ref{sub:gen_bip}. The goal is to observe the behavior of the minimum gap and to confirm the exponentially closing gap. We then discuss whether or not these gaps lead to a computational inefficiency of QA and moderate the term AC by looking at the more mathematical definition of \cite{Braida_2021}.

\paragraph{Minimum gap study:} 
% Let us first investigate the minimum gap required to support the theoretical framework discussed in . 
Let us first show that the value of the minimum gap supports the theoretical results derived in Section \ref{sec:3} and Subsection \ref{sub:gen_bip}.
To compute this quantity for large graphs, we use the SciPy library \cite{scipy} and its optimized method, scipy.sparse.linalg.eigs, for matrices with a sparse representation. Our Hamiltonians have a sparse representation in the Pauli basis, enabling us to compute the minimum gap for graphs with up to 28 nodes.

To satisfy the conditions required for our application, we fix one node of the graph to lift the standard MaxCut symmetry. Specifically, we fix one node of the $K_{l,l}$ part on the left (L) side of the partition. We consider the family of graphs $G_{rk}$ with the same structure as in the previous section, where we fix $l=3$ and assume $k_1=k_2=k$. Therefore, we can vary two parameters (Figure \ref{fig:Cf_Grk} shows the schematic energy landscape of $H_1$ for $G_{rk}$ :

\begin{enumerate}
\item[-] increasing $r$ increases the distance between $G_{loc}$ and the ground state in the hypercube, as all the $K_{r,r}$ part needs to be flipped (fixing one node in the $K_{l,l}$ part blocks the possibility to flip this part entirely),
\item[-] increasing $k$ creates a larger $G_{loc}$, resulting in a larger local minimum that is not linked to the ground state, but also increases the two other parts of $G_{loc}$ connected to it.
\end{enumerate} 

\begin{figure}
    \centering
    \includegraphics[scale=0.5]{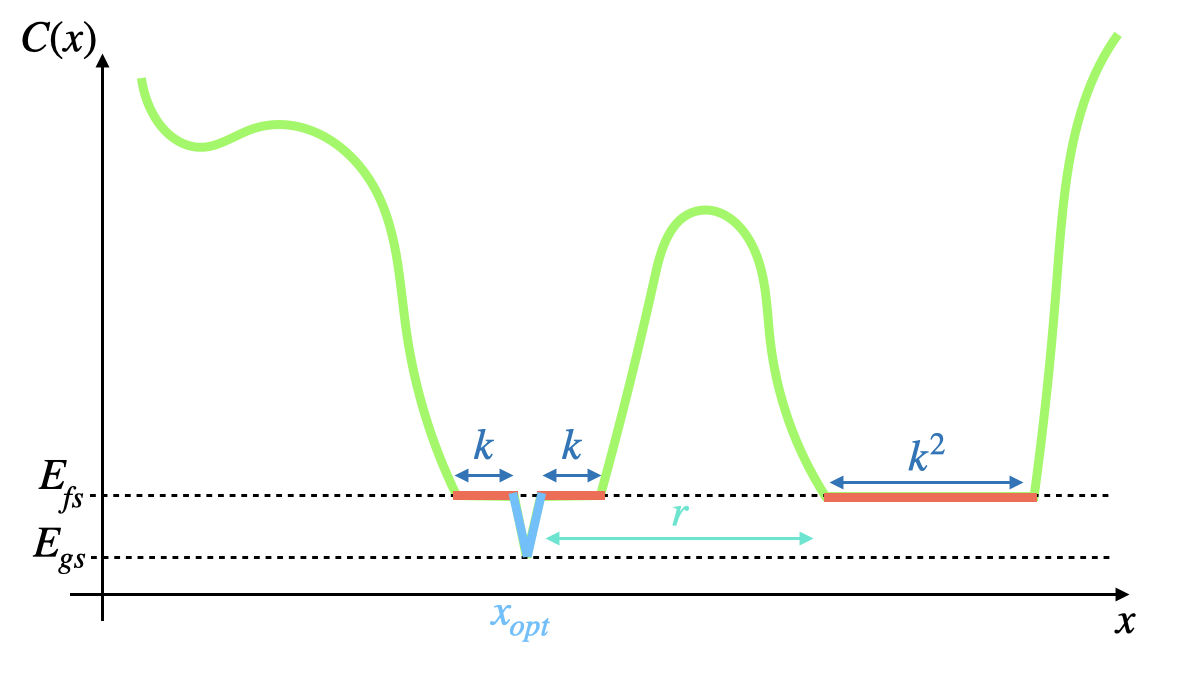}
    \caption{Schematic energy landscape of the MaxCut function on an instance $G_{rk}$ and how $r$ and $k$ affect it.}
    \label{fig:Cf_Grk}
\end{figure}

We denote $\Delta_{rk}(s)$ as the difference between the two lowest instantaneous eigenvalues of $H(s)$ associated with $G_{rk}$, i.e., the spectral gap of the time-dependent Hamiltonian. We plot these gaps in Figure~\ref{fig:delta_rk} (a) by varying $r$ and $k$. Specifically, we observe that increasing $r$ by 1 divides the gap by 2. To illustrate this, we also plot Figure~\ref{fig:delta_rk} (b) the minimum gap of $\Delta_{rk}$ for $k=2$ against $r$. We fit this curve with an exponentially decreasing function of $r$. When $k$ is fixed, it is straightforward to see that $r \simeq \frac{n}{2}$.

\begin{figure}[ht]
    \centering
    $\begin{array}{cc}
       \includegraphics[scale=0.41]{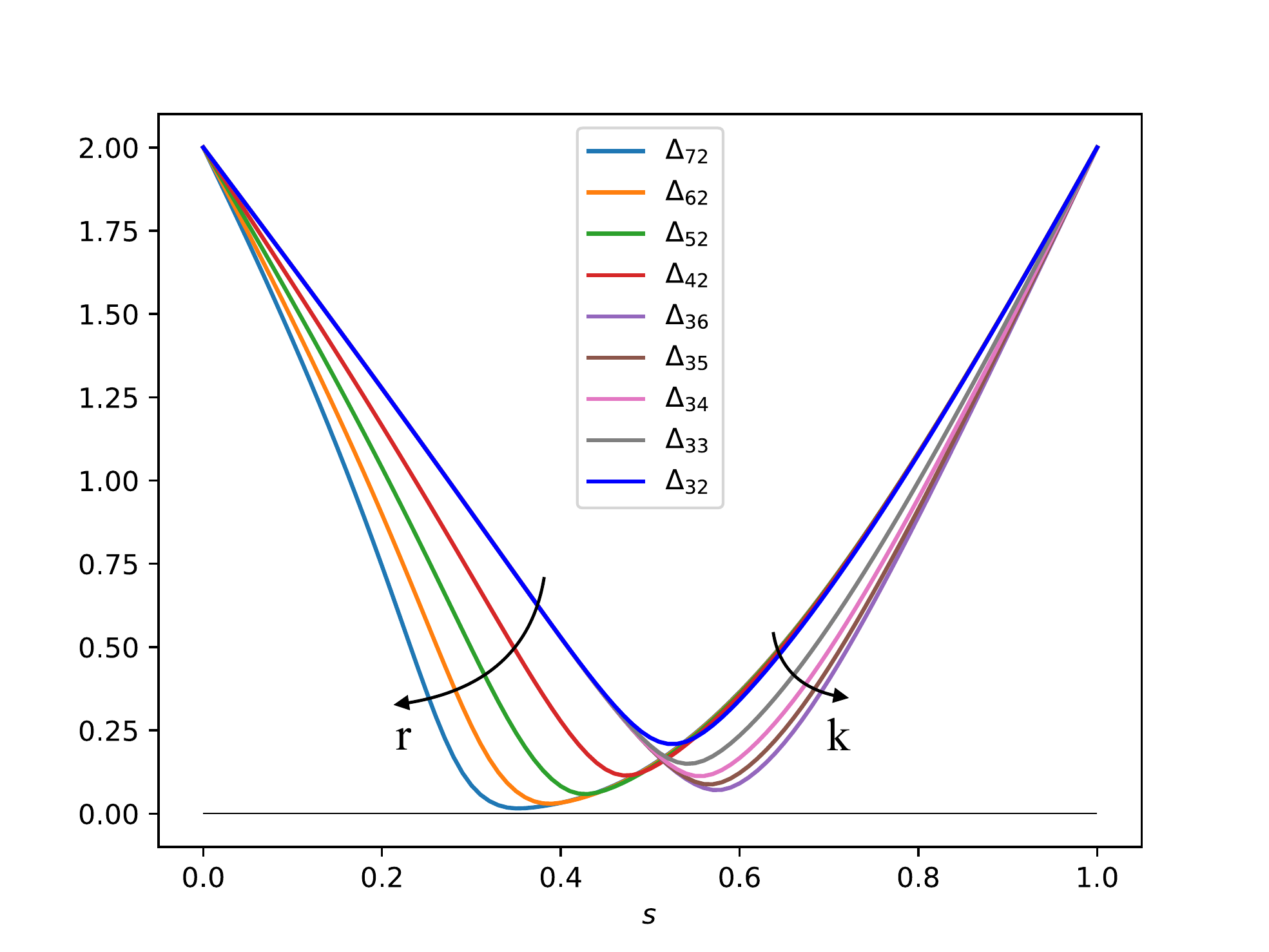}  & \includegraphics[scale=0.5]{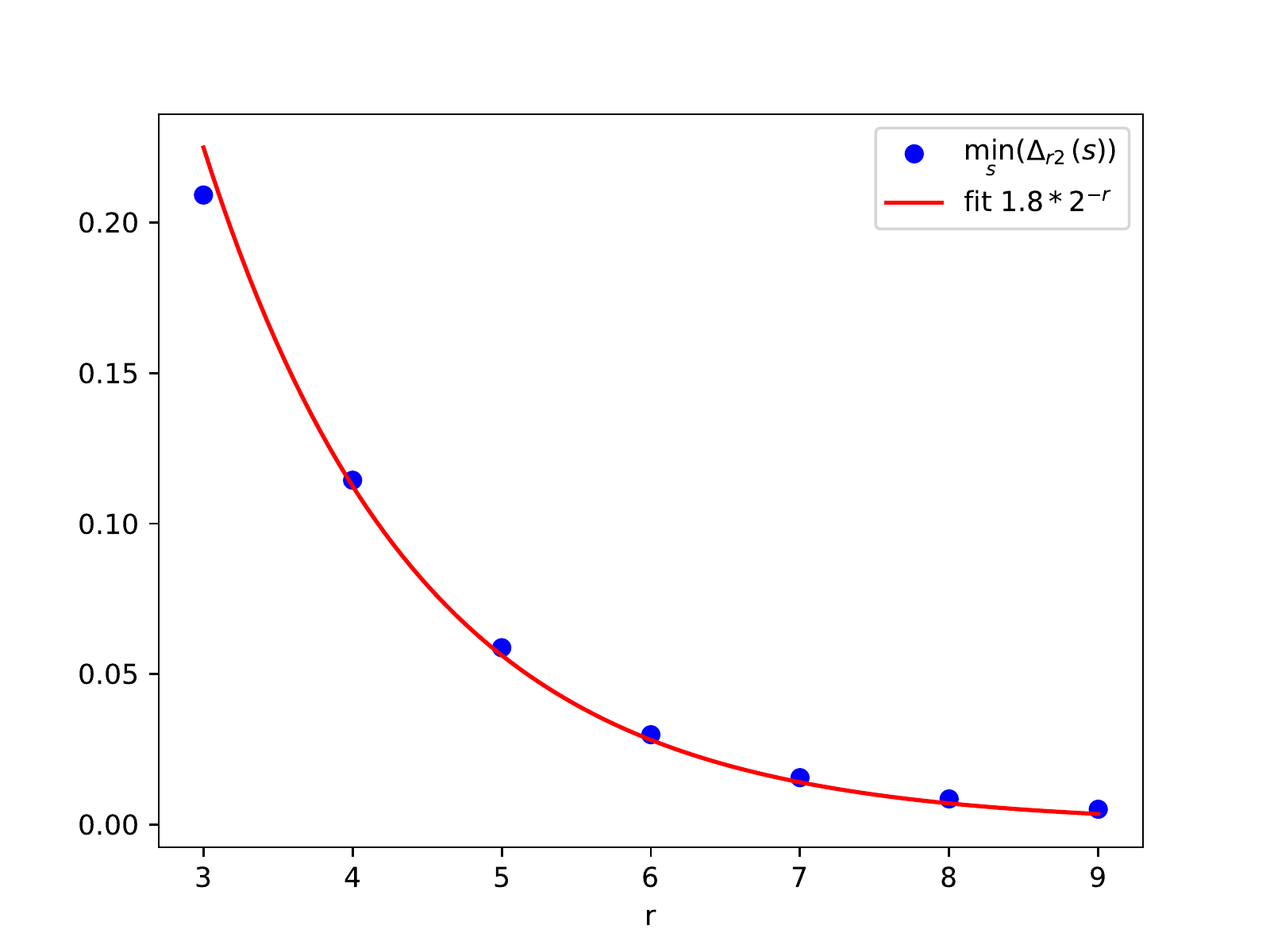} \\
        (a) & (b)
    \end{array}$
    
    \caption{(a) Evolution of the spectral gap $\Delta_{rk}(s)$ and (b) Minimum gap of $\Delta_{r2}$ for $r$ going from 3 to 9. It fits an exponentially decreasing tendency.}
    \label{fig:delta_rk}
\end{figure}
Figure~\ref{fig:delta_rk} supports the main theorem in Section~\ref{sec:3} and the construction in Section~\ref{sec:4}. The distance to the ground state appears to play a major role in the minimum gap compared to the size of $G_{loc}$. Remember that $G_{loc}$ has three components and two of them are linked to the ground state while the other one is a lattice far from the ground state. Increasing $k$ also increases the width around the ground state, making it easier to reach than if it were isolated while increasing $r$ has no impact on $G_{loc}$. \\

\noindent Typically, it is assumed that an exponentially closing gap implies the failure of QA \cite{altshuler2010anderson}. In the next paragraph, we investigate the probability of measuring the ground state at the end of a QA evolution after a time $t_{\max}$ and discuss about AC definition which opens a new question on the computational efficiency of QA.

\paragraph{Discussions about AC and QA failure:} Now that we have established the exponentially small gaps for the graph $G_{rk}$ when $r$ is increasing, we can wonder if it can be deduced that QA is inefficient to solve those instances, as this is the usual deduction from small gaps. In Figure \ref{fig:pr2}, we observe the probability $p_{rk}$ of measuring the ground state at the end of a quantum annealing (QA) evolution for different instances of $G_{rk}$ as a function of $t_{\max}$. This plot was obtained using the AnalogQPU of the Atos' quantum software. Surprisingly, the probability seems to reach the value around 0.5 faster than expected, meaning in a time $t_{\max}$ that does not appear to depend too much on the size of the graph. This observation is not a contradiction of the adiabatic theorem, as it will certainly converge to 1 in an exponentially long runtime. It could also be just a scale illusion: for much larger graphs, the probability might stay at 0 for a longer time than observed here, but this is not what the point below suggests. However, it raises questions about the effectiveness of QA in practical applications even when exponentially small gaps are present. 

\begin{figure}[ht]
    \centering
    \includegraphics[scale=0.6]{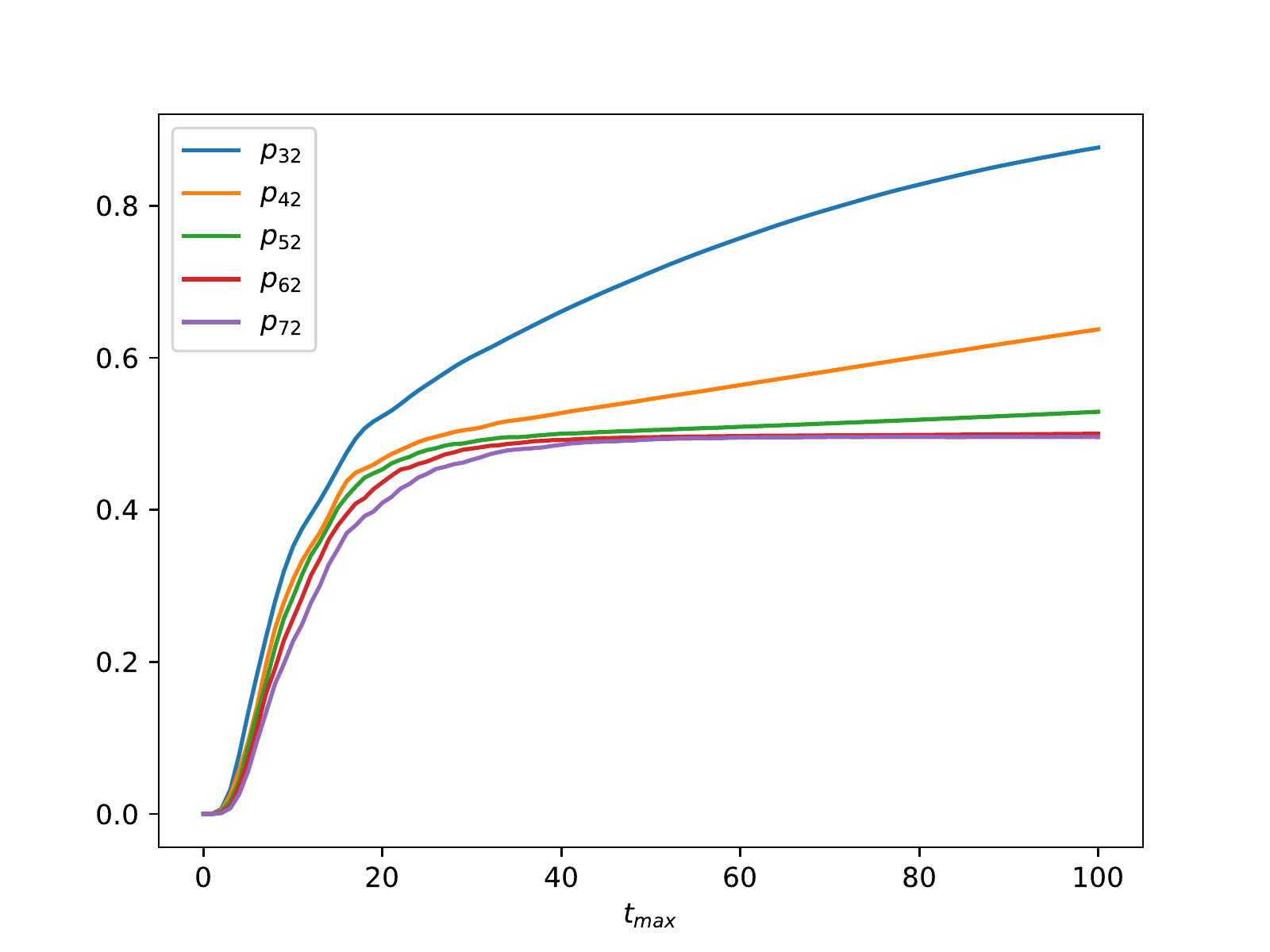}
    \caption{Probability of measuring the ground state after a time $t_{\max}$ for instances with $k=2$ and $r \in [3,7]$.}
    \label{fig:pr2}
\end{figure}

\noindent The observed gaps in Figure \ref{fig:delta_rk} exhibit an exponentially closing behavior, which is a signature of the AC phenomenon we are looking at. However, the computational complexity does not seem to be affected, in the sense that a constant probability to obtain the optimal solution is reached in a time that does not seem to depend too much on the graph size. We can notice in Figure \ref{fig:delta_rk} (a) that the trend of the gaps appears to be softer compared to other observed ACs \cite{bapst2012quantum}, indicating a smoother transition. To address this observation, \cite{Braida_2021} proposed a more formal definition of anti-crossings that involves a new set of quantities. Let $g_0(s)=|\langle E_0(s)|GS\rangle|^2$ and $g_1(s)=|\langle E_1(s)|GS\rangle|^2$ be the overlap squared of the instantaneous eigenstate (zeroth and first respectively) of $H(s)$ with the ground state $|GS\rangle$ of $H_1$. Typically, at anti-crossing point, these curves undergo a harsh exchange of position. If $g_0(s)$ smoothly increases toward 1, it is not an AC according to this definition. For the graph $G_{rk}$, the conditions given in this formal definition do not seem to be fully satisfied, as the plots in Figure \ref{fig:gs_behave} show. On the left, an example of behavior of $g_0$ and $g_1$ when AC happens, the curves experience an almost discontinuity at AC point, changing the position of $g_0$ and $g_1$. On the middle and right plots, $g_0$ and $g_1$ for instances $G_{32}$ and $G_{72}$ respectively. In the $G_{32}$ case, $g_1$ starts to become bigger than $g_0$ but it produces only a little bump and $g_0$ has a smooth growth toward 1.  One could think that this phenomenon is due to the small size of the instance, and that by considering larger instances but with very small gaps, we would observe a ``typical'' AC behaviour. However, on the $G_{72}$ case, where the size increases and the gap decreases, this bump totally disappears and we can only attest a smooth growth of $g_0$. This observation indicates the opposite of an AC behavior leading to an efficient QA evolution to solve these instances.
%The closing gaps as AC signature suggest the inverse, i.e. the curves of $g_j$ should converge toward a typical behavior as the size increase\sout{, while the minimum gap decreases the AC should be stronger.}\itodo{Phrase précédente à améliorer ?} Here we observe the opposite with a $g_0$ curve that smoothly increases toward 1 indicating no difficulty from a QA point of view to solve the instance. 
This raises the question of whether every exponentially closing gap necessarily leads to a failure of QA, or if AC is a particular event that creates an exponentially closing gap leading to a complete leak of the probability distribution into higher energy levels.

\begin{figure}
    \centering
    \includegraphics[scale=0.35]{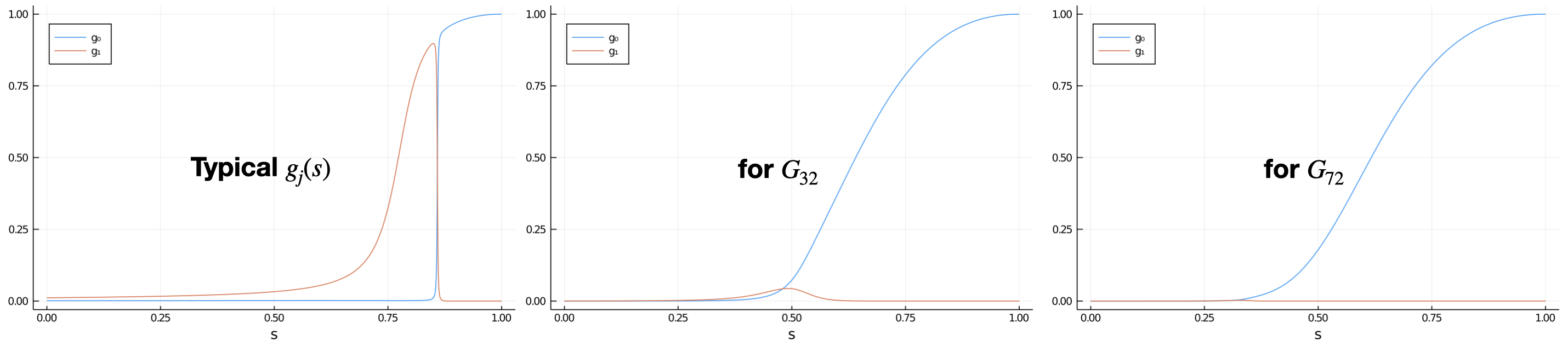}
    \caption{Evolution of $g_0(s)$ in blue and $g_1(s)$ in red for graph $G_{72}$ (right), $G_{32}$ (middle) and a typical behavior (left) during an AC like in \cite{Braida_2021}}
    \label{fig:gs_behave}
\end{figure}

\section{Conclusion}

In conclusion, in this work, we prove a new theorem showing a condition on the occurrence of anti-crossings during a quantum annealing process. The signature of AC we consider is the presence of exponentially closing gaps. Using a first-order perturbative analysis of the evolution at the beginning on the initial ground energy of $H_0$ and at the end on the non-degenerated ground energy and degenerated first eigenenergy of $H_1$, we manage to articulate these expressions together to derive a general condition on the occurrence of AC. In particular, if $\lambda_0(\text{loc})>n\alpha_T$, an AC occurs, where $\lambda_0(\text{loc})$ is the principal eigenvalue of the graph $G_{loc}$ which is the graph induced by the degenerated states of the first eigenspace of $H_1$ in $H_0$. In other words, $G_{loc}$ is the restriction of $H_0$ to the states that belong to the first excited space of $H_1$. $\alpha_T$ is a parameter that depends only on the target Hamiltonian, i.e. only on the problem we want to solve. 

We then apply this theorem to the MaxCut problem and we show that for regular bipartite graphs, the condition for AC to occur is never fulfilled meaning that the gap stays ``large'' for these instances. This means that it is efficient to solve MaxCut on regular bipartite graphs with quantum annealing. As far as we know, it is the first proof of efficiency for this problem on any class of graphs. 

Eventually, by removing the regularity assumption, we manage to create highly irregular bipartite graphs that satisfied the AC conditions. To support the theoretical development, we numerically investigate the size of the gap while increasing the size of the graph. We show that the minimum gap has an exponentially decreasing fit. Surprisingly, while this usually implies the inefficiency of the annealing process to solve those instances, we observe that the final probability of measuring the ground state at the end of the process seems to reach a constant value of 0.5 independently of the graph size. This means that despite an exponentially closing gap, the computational complexity to solve the instances is not affected. We further investigate the AC phenomenon in these cases by using a more formal definition of AC and conclude that our instances with small gaps do not meet this latter AC description. This opens the question of whether the presence of an exponentially closing gap necessarily entails inefficient annealing, or if the true marker of inefficiency is the presence of an AC as defined in \cite{Braida_2021}.
% or if only some of it that meets other conditions like the one that define an AC in \cite{Braida_2021}.

Overall, our study provides new insights into the efficiency of quantum annealing for solving optimization problems, particularly the MaxCut problem, and highlights the importance of considering the occurrence of ACs during the annealing process.

\bibliographystyle{unsrt}
\bibliography{biblio}

\begin{thebibliography}{10}

\bibitem{farhi2000quantum}
Edward Farhi, Jeffrey Goldstone, Sam Gutmann, and Michael Sipser.
\newblock Quantum computation by adiabatic evolution.
\newblock {\em arXiv preprint quant-ph/0001106}, 2000.

\bibitem{albash2018adiabatic}
Tameem Albash and Daniel~A Lidar.
\newblock Adiabatic quantum computation.
\newblock {\em Reviews of Modern Physics}, 90(1):015002, 2018.

\bibitem{farhi2014quantum}
Edward Farhi, Jeffrey Goldstone, and Sam Gutmann.
\newblock A quantum approximate optimization algorithm.
\newblock {\em arXiv preprint arXiv:1411.4028}, 2014.

\bibitem{wilkinson1989statistics}
Michael Wilkinson.
\newblock Statistics of multiple avoided crossings.
\newblock {\em Journal of Physics A: Mathematical and General}, 22(14):2795,
  1989.

\bibitem{Amin_2009}
M.~H.~S. Amin and V.~Choi.
\newblock First-order quantum phase transition in adiabatic quantum
  computation.
\newblock {\em Phys. Rev. A}, 80:062326, Dec 2009.

\bibitem{altshuler2010anderson}
Boris Altshuler, Hari Krovi, and J{\'e}r{\'e}mie Roland.
\newblock Anderson localization makes adiabatic quantum optimization fail.
\newblock {\em Proceedings of the National Academy of Sciences},
  107(28):12446--12450, 2010.

\bibitem{choi2020effects}
Vicky Choi.
\newblock The effects of the problem hamiltonian parameters on the minimum
  spectral gap in adiabatic quantum optimization.
\newblock {\em Quantum Information Processing}, 19(3):90, 2020.

\bibitem{Braida_2021}
Arthur Braida and Simon Martiel.
\newblock Anti-crossings and spectral gap during quantum adiabatic evolution.
\newblock {\em Quantum Information Processing}, 20(8), aug 2021.

\bibitem{feinstein2022effects}
Natasha Feinstein, Louis Fry-Bouriaux, Sougato Bose, and PA~Warburton.
\newblock Effects of xx-catalysts on quantum annealing spectra with
  perturbative crossings.
\newblock {\em arXiv preprint arXiv:2203.06779}, 2022.

\bibitem{Werner:2023zsa}
Matthias Werner, Artur Garc\'\i{}a-S\'aez, and Marta~P. Estarellas.
\newblock {Bounding first-order quantum phase transitions in adiabatic quantum
  computing}.
\newblock {\em arXiv preprint arXiv:2301.13861}, 1 2023.

\bibitem{Barahona1988}
Francisco Barahona, Martin Gr\"{o}tschel, Michael J\"{u}nger, and Gerhard
  Reinelt.
\newblock An application of combinatorial optimization to statistical physics
  and circuit layout design.
\newblock {\em Oper. Res.}, 36(3):493–513, jun 1988.

\bibitem{zwiebach:2018}
Barton Zwiebach.
\newblock Chapter 1: Non-degenerate and degenerate perturbation theory.
\newblock In {\em QUANTUM PHYSICS III ---MIT Course}. MIT OpenCourseWare, 2018.
\newblock {MIT OpenCourseWare}.

\bibitem{graphEig}
Yueheng Zhang.
\newblock On the principal eigenvector of a graph.
\newblock {\em arXiv preprint arXiv:2107.14421}, 2021.

\bibitem{scipy}
Eric Jones, Travis Oliphant, Pearu Peterson, et~al.
\newblock {SciPy}: Open source scientific tools for {Python}, 2001--.

\bibitem{bapst2012quantum}
Victor Bapst and Guilhem Semerjian.
\newblock On quantum mean-field models and their quantum annealing.
\newblock {\em Journal of Statistical Mechanics: Theory and Experiment},
  2012(06):P06007, 2012.

\bibitem{cvetkovic1980spectra}
D.M. Cvetkovic, D.M. Cvetkovi{\'c}, M.~Doob, and H.~Sachs.
\newblock {\em Spectra of Graphs: Theory and Application}.
\newblock Pure and applied mathematics : a series of monographs and textbooks.
  Academic Press, 1980.

\end{thebibliography}

\appendix

\section{Validation of perturbative expansion}
\label{app:val}
We discuss here the validation of this expansion at first order. We need to look at the second order term and compare it to the first or 0$^{th}$ order term. 

\paragraph{Delocalized state expansion:} The eigen basis of the initial Hamiltonian $H_0$ can be written as $$|E_b\rangle = \frac{1}{\sqrt{2^n}} \sum_{x\in\{0,1\}^n} (-1)^{b\cdot x}|x\rangle$$ where $b$ is an $n-$bitstring and $\cdot$ stands for the scalar product over $\mathbb{F}_2^n$. There are $n+1$ differents eigen levels where the $k^{th}$ eigenspace has degeneracy $\binom{n}{k}$ and correspond to eigenstates with bitstring $b$ of hamming weight $|b|=k$ and eigen value $E_b^I=-n+2|b|$ (see \cite{cvetkovic1980spectra} for more details). With this notation, we can write $|\psi_0\rangle$ as $|E_{00...00}\rangle$. We are interested in $$E_{0}^{(2)}=\sum_{b\neq 00...0} \frac{|\langle E_b |H_1|\psi_0\rangle|^2}{E_0^I-E_b^I}$$ For MaxCut problem on graph $G$, we know that $\langle E_b |H_1|\psi_0\rangle = -1/2$ if and only if $G_b$ is exactly one edge. $G_b$ is the graph induced by the node $i$ where $b_i=1$. Therefore $E_{0}^{(2)}=-\frac{|E(G)|}{16}$. We have $E_{0}^{(1)}=\langle H_1 \rangle_0 \simeq -\frac{|E(G)|}{2}$ so $\frac{|E_{0}^{(2)}|}{|E_{0}^{(1)}|}=\frac{1}{8}<1$. 

\paragraph{Ground state expansion:} The eigenbasis of the final Hamiltonian $H_1$ is the canonical basis of the bitstring $|x\rangle$ with energy $E_x$, and we named $|GS\rangle$, the bitstrings corresponding to the ground state with energy $E_{gs}$. It follows that the second order term is
$$E_{gs}^{(2)} =\sum_{x \in \{0,1\}^n} \frac{|\langle x|H_0|GS\rangle|^2}{E_{gs}-E_{x}}$$ where $|\langle x|H_0|GS\rangle|=1$ if and only if the bitstring $x$ is at exactly one bitflip from the bitstring $GS$. We can rewrite it like 
$$E_{gs}^{(2)} =\sum_{x \underset{H_0}{\sim} GS} \frac{1}{E_{gs}-E_{x}}
$$ For MaxCut problem on $d-$regular bipartite graph, we can further simplify. Indeed, from the ground state, flipping one bit gives an energy state $|x\rangle$ of exactly $E_x=E_{gs}+d$. So we end up with $E_{gs}^{(2)}=-\frac{n}{d}$. We have $E_{gs}^{(1)}=0$ and $E_{gs}^{(0)}=E_{gs}=\frac{dn}{2}$ which leads to $\frac{|E_{gs}^{(2)}|}{|E_{gs}^{(0)}|}=\frac{2}{d^2}<1$. For $d=4$ we have the same value as for the delocalized state.

\paragraph{Local minima expansion:} We work in the same basis than gor the latter expansion and we look at 
$$ E_{fs}^{(2)}=\sum_{x \notin V(G_{loc})} \frac{|\langle x|H_0|FS,0\rangle|^2}{E_{fs}-E_x} \leq \sum_{x \notin V(G_{loc})} \sum_{y \underset{H_0}{\sim} x}\frac{|\langle y|FS,0\rangle|^2}{E_{fs}-E_x}$$ The size of this double sum is the number of connection $G_{loc}$ has with the whole hypercube, i.e. $|\partial G_{loc}|$. The term $|\langle y|FS,0\rangle|^2$ is large when the degree of node $y$ in $G_{loc}$ is large so with less occurrence in the above double sum. In average, when a graph is regular its vector coordinate value of the largest eigenvalue is $\frac{1}{\sqrt{|V(G_{loc})|}}$. By introducing the conductance of the subgraph $G_{loc}$ as $\phi(\text{loc})=\frac{|\partial G_{loc}|}{|V(G_{loc})|}$, we can upper bound the second order term with 
$$ |E_{fs}^{(2)}| \leq \phi(\text{loc}) \frac{1}{\min_x |E_{fs}-E_x|}$$ We know that $|E_{fs}^{(1)}| =\lambda_0(\text{loc})\geq d_{\text{avg}}(\text{loc})=n-\phi(\text{loc})$. So the ratio we need to check is $\frac{\phi(\text{loc})}{n-\phi(\text{loc})}\frac{1}{\min_x |E_{fs}-E_x|}$ which smaller if $G_{loc}$ is neighboring high energy states.

\section{Undefined cases of d-regular bipartite graphs}
\label{app:case}
\subsection{Case d=2}
\label{ssec:cycle}
\paragraph{Cycle case (even): } Looking at the specific case of the even cycle, we see that it creates a large $G_{loc}$, see Figure~\ref{fig:GlocCycle}. 
\begin{figure}[ht]
    \centering
    \includegraphics[scale=0.3]{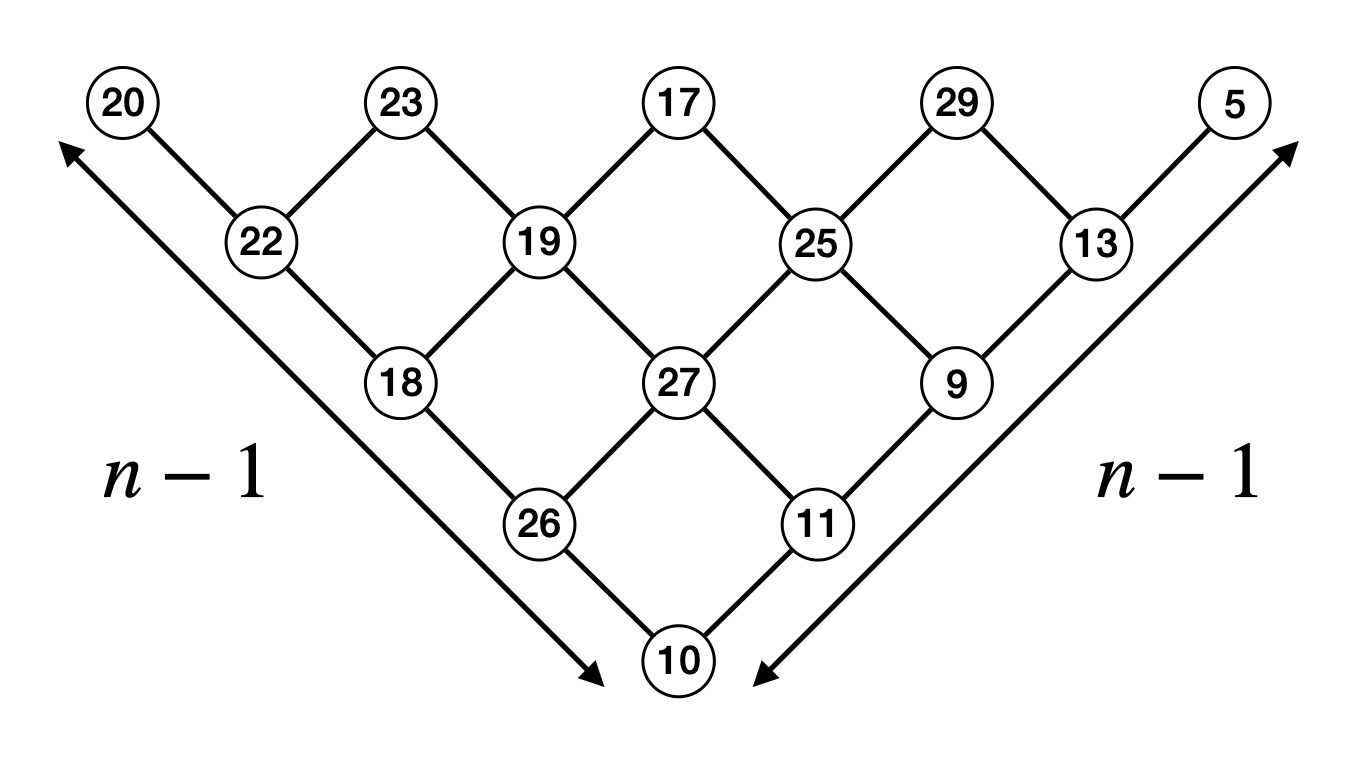}
    \caption{$G_{loc}$ of a cycle of size $n=6$}
    \label{fig:GlocCycle}
\end{figure}
We can easily evaluate the average and maximum degree of this graph as :
\begin{align*}
    \deg_{\max} (\text{loc})&=4 \\
    \deg_{\text{avg}} (\text{loc}) &= 4 \frac{n-2}{n} = 4(1-\frac{2}{n})
\end{align*}

These values bring the cycle in the UNDEFINED regime. However, we can expect that QA will easily work with a MaxCut on an even cycle because its $G_{loc}$ is highly connected to the ground state. Figure \ref{fig:gh0_cycle} shows how $G_{loc}$ (which is the one in figure \ref{fig:GlocCycle}) is linked to te ground state (blue edges). More precisely, there are $n-1$ connection with the ground state in a $(n-1)$-regular graph. This means that there is no potential barrier to overcome going from $G_{loc}$ to the GS.

Another justification is to directly look at the main theorem which says that no AC occurs if $\lambda_0 < 4$, where $\lambda_0$ is the largest eigenvalue of $G_{loc}$. We know that $\lambda_0=\deg_{\max} (\text{loc})$ if and only if $G_{loc}$ is $\deg_{\max} (\text{loc})$-regular, otherwise $\lambda_0<\deg_{\max} (\text{loc})$. So we are in the no-AC regime.

\subsection{Case d=4}
By construction, in the case where $d=4$, there is one possible configuration in a 4-regular graph that brings its $G_{loc}$ in the UNDEFINED regime. It can be artificially scale up as follow :

\begin{figure}[ht]
    \centering
    \includegraphics[scale=0.6]{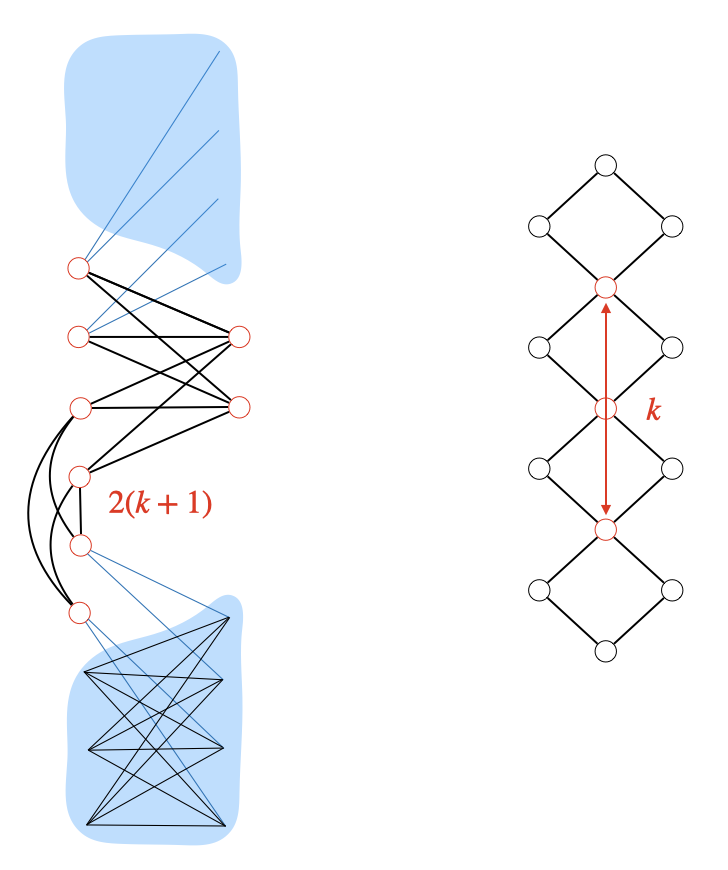}
    \caption{(left) 4-regular bipartite graph in one of its first excited state configuration and (right) $G_{loc}$ where we disregarded the isolated node. Written in red, the number of red nodes ($k=3$). In blue, a part of the graph that complete the graph in a 4-regular one.}
    \label{fig:bipartited4}
\end{figure}

We can easily derive the maximum and average degree of $G_{loc}$:
\begin{align*}
    \deg_{\max} (\text{loc})&=4 \\
    \deg_{\text{avg}} (\text{loc}) &=  \frac{8(k+1)}{3k+4} = 2+\frac{2k}{3k+4}
\end{align*}
where $k$ is a parameter to construct the graph. $G_{loc}$ is not connected to the ground state, so one can imagine that this will produce a potential barrier that creates an AC. But as one can see, the average degree of $G_{loc}$ only tends to $2+2/3 < 4$, which is far from the AC appearance condition. A similar argument from the case $d=2$ can be applied here when using directly the technical theorem with $\lambda_0$.

\end{document}